\documentclass[11pt, reqno]{amsart}

\usepackage{amsthm}
\usepackage{amsmath}
\usepackage{amssymb}
\usepackage{mathrsfs}
\usepackage{mathtools}
\usepackage{bbm}

\usepackage[open,openlevel=1]{bookmark}
\usepackage{xcolor}
\usepackage{float}
\hypersetup{
colorlinks,
linkcolor={red!50!black},
citecolor={blue!50!black},
urlcolor={blue!80!black}
}

\usepackage{thmtools, thm-restate}  

\declaretheoremstyle[
spaceabove=6pt, spacebelow=6pt,
headfont=\normalfont\bfseries,
notefont=\mdseries, notebraces={(}{)},
bodyfont=\normalfont,
postheadspace=1em,
qed=$\blacksquare$
]{examplestyle}

\declaretheoremstyle[
spaceabove=6pt, spacebelow=6pt,
headfont=\normalfont\bfseries,
notefont=\mdseries, notebraces={(}{)},
bodyfont=\itshape,
postheadspace=1em
]{theorem}

\declaretheoremstyle[
spaceabove=6pt, spacebelow=6pt,
headfont=\normalfont\bfseries,
notefont=\mdseries, notebraces={(}{)},
bodyfont=\normalfont,
postheadspace=1em
]{assumption}

\declaretheoremstyle[
spaceabove=4pt, spacebelow=4pt,
headfont=\itshape\bfseries,
notefont=\mdseries, notebraces={(}{)},
bodyfont=\itshape,
postheadspace=0.2em,
qed=\qedsymbol
]{remark}

\declaretheorem[style=theorem]{theorem}
\declaretheorem[style=theorem, numbered=no, name=Theorem]{theorem*}
\declaretheorem[style=remark,name=Remark]{remark}

\declaretheorem[style=assumption,name=Assumption]{assumption}

\declaretheorem[style=theorem,name=Definition]{definition}
\declaretheorem[style=definition, name=Lemma]{lemma}

\declaretheorem[style=theorem, name=Proposition]{proposition}
\declaretheorem[numbered=no,style=definition,name=Question]{question*}  
\declaretheorem[style=definition,name=Definition, numbered=no]{definition*}  

\declaretheorem[style=theorem,name=Lemma, numberwithin=section]{lemmasec}

\declaretheorem[style=definition, name=Definition, numberwithin=section]{defsec}

\declaretheorem[style=examplestyle]{example}
\renewcommand\thmcontinues[1]{continued}

\newtheorem{thm}{Assumption}

\newcommand{\nature}{\mathbb{N}}
\newcommand{\real}{\mathbb{R}}

\newcommand{\norm}[1]{\left\Vert #1\right\Vert }

\newcommand{\indicator}{\mathbbm{1}}

\newcommand{\ud}{\mathrm{d}}
\newcommand{\prob}{\mathbb{P}}
\newcommand{\expt}{\mathbb{E}}

\newcommand{\co}{\text{\normalfont co}}									
\newcommand{\clco}{\overline{\text{\normalfont co}}}		
\newcommand{\cl}{\text{\normalfont cl}}									
\newcommand{\interior}{\text{\normalfont int}}								
\newcommand{\ri}{\text{\normalfont ri}}

\newcounter{steps}

\DeclareMathOperator*{\argmax}{arg\,max}

\usepackage[backend=biber, 
    style=authoryear, 
	  date=year,
		uniquename=false,
		uniquelist=false,
		doi=false,    
    isbn=false,   
    url=false,    
    eprint=false  
]{biblatex}

\AtEveryBibitem{
    \clearfield{note}      
    \clearfield{addendum}  
    \clearlist{publisher}
    \clearlist{location}
    \clearfield{issn}
    \clearfield{isbn}
}
\renewbibmacro{in:}{}

\let\cite\textcite
\usepackage[nodisplayskipstretch]{setspace}
\usepackage{geometry}
\usepackage{enumitem}
\setenumerate[1]{label=(\roman*)}
\setenumerate[2]{label=\normalfont{(\alph*)}}

\begin{document}

\title[]{Identification and Counterfactual Analysis in Incomplete Models with Support and Moment Restrictions}


\author[]{Lixiong Li}
\address{Johns Hopkins University}
\email{lixiong.li@jhu.edu}
\thanks{This paper supersedes my job market paper, originally circulated under the title ``Identification of Structural and Counterfactual Parameters in a Large Class of Structural Econometric Models.'' I am deeply grateful to Marc Henry, Keisuke Hirano, and Joris Pinkse for their invaluable guidance and support during that project. I also thank Michael Gechter, Paul Grieco, Patrik Guggenberger, Charles Murry, Mark Roberts, Karl Schurter, and Neil Wallace for helpful discussions and suggestions. Special thanks to Ismael Mourifi\'e, Timothy Christensen and Francesca Molinari for encouraging me to revise the paper into its current form. All remaining errors are my own. 
\vspace{1em} \\ Address: Wyman Park Building 5th Floor, 3100 Wyman Park Drive, Baltimore, MD 21211 \\  Email: lixiong.li@jhu.edu \\
This version: \today
}


\date{}

\dedicatory{Johns Hopkins University}

\begin{abstract}

This paper develops a unified identification framework for counterfactual analysis in incomplete models characterized by support and moment restrictions. I demonstrate that identifying structural parameters and conducting counterfactual analyses are isomorphic tasks. By embedding counterfactual restrictions within an augmented structural model specification, this approach bypasses the conventional "estimate-then-simulate" workflow and the need to simulate outcomes from models with set predictions.  To make this approach operational, I extend sharp identification results for the support-function approach beyond the integrable boundedness condition that is imposed in sharp random-set characterizations but may be violated in economically relevant counterfactual analyses.  Under minimal regularity conditions, I prove that the support-function approach remains sharp for the \emph{moment closure} of the identified set. Furthermore, I introduce an irreducibility condition requiring all support implications to be made explicit. I show that for irreducible models, the identified set and its moment closure are statistically indistinguishable in finite samples. Together, these results justify using support-function methods in counterfactual settings where traditional sharpness fails and clarify the distinct roles of support and moment restrictions in empirical practice.

\vspace{20pt}

\end{abstract}

\maketitle
\newpage

\section*{Introduction}

Counterfactual analysis is central to applied work throughout economics. A common approach is to specify a structural model, estimate its parameters, and then use the estimated model to predict behavior under policies or environments that are not observed in the data, such as changes in market size, mergers, entry regulation, or price and tax interventions. In many empirically relevant settings, however, the structural model does not deliver a unique prediction for the outcome. Models that generate set-valued predictions are commonly referred to as \emph{incomplete models}: conditional on observables, latent variables, and parameters, the model implies a set of outcomes rather than a single outcome.  Incompleteness can arise for a variety of reasons. A leading example is a game of complete information with a pure-strategy equilibrium concept, in which multiple equilibria may exist (e.g., \cite{jovanovic_observable_1989, tamer_incomplete_2003}). Other examples include models of choice with limited attention (\cite{barseghyan_heterogeneous_2021-1}), discrete choice models with endogeneity (\cite{chesher_instrumental_2013}), and auction models (\cite{haile_inference_2003}).  \cite{molinari_microeconometrics_2020} and \cite{chesher_instrumental_2020} provide comprehensive surveys of incomplete structural models and related partial-identification methods.

Conducting counterfactual analysis in incomplete models raises both computational and statistical challenges. On the
computational side, the conventional workflow of estimating a point-identified structure and simulating counterfactual
outcomes is often computationally challenging or infeasible, because counterfactual outcomes are not uniquely determined even
conditional on primitives and parameters. On the statistical side, as illustrated later, working with economically important counterfactual targets, such as welfare, profits, and surplus, often violates key regularity conditions required for sharp identification using random set methods. Most notably, the assumption of integrable boundedness routinely fails in these counterfactual exercises.

This paper studies counterfactual analysis and related identification issues in a broad class of models characterized by
(\emph{i}) a restriction on the joint support of the random variables in the model and (\emph{ii}) a collection of moment
restrictions that may involve both observed and latent variables. The support restriction is implied by economic theory and
typically encodes behavioral assumptions, equilibrium concepts, and structural features of the environment. The moment
restrictions allow the researcher to restrict latent variables without specifying their full distribution. This
support-and-moments framework nests a wide range of empirically relevant structural models.

This paper makes two main contributions. First, it develops a uniform framework for counterfactual analysis in incomplete models with set-valued predictions. The key insight is that, in this setting, identifying counterfactual parameters is isomorphic to identifying structural parameters. A counterfactual exercise introduces additional restrictions: the analyst must specify how counterfactual outcomes and possibly counterfactual latent variables relate to the primitives of the baseline structural model, and define the counterfactual parameter of interest through moment conditions that may involve both observed and latent variables. An essential observation is that these counterfactual restrictions can be stacked with the baseline structural restrictions to form a single augmented support-and-moments model. In the augmented model, counterfactual parameters can be treated on the same footing as structural parameters: identification reduces to characterizing the set of parameter values consistent with the combined support and moment restrictions.  When only a counterfactual parameter is of interest, inference can then be conducted by applying standard subvector inference procedures within the augmented model using existing methods for partially identified models.
This unified formulation avoids the need to simulate counterfactual outcomes from a set-predicting structural model and clarifies how counterfactual targets inherit partial identification from the underlying structural environment.

The paper's second contribution is to extend sharp identification results for the support-and-moments framework beyond the integrable boundedness condition that underlies existing results based on random set theory (see, e.g., \cite{ekeland_optimal_2010,beresteanu_sharp_2011}). In many counterfactual exercises, integrable boundedness fails naturally. For example, equilibrium conditions often impose only one-sided restrictions on latent variables, so the random sets associated with counterfactual objects such as profits or welfare may be unbounded, thereby violating the integrable boundedness condition. When this occurs, the standard support-function approach need not characterize the identified set itself. Nevertheless, I show that it still captures the effective limit of what can be learned from the data. Achieving this interpretation requires an \emph{irreducibility} condition: the model must be formulated so that support restrictions are imposed explicitly, while the moment restrictions do not contain additional support implications. This requirement is easy to overlook because any support restriction can always be rewritten as a moment condition. I show, however, that support restrictions and moment restrictions play fundamentally different roles in identification analysis, and that mixing the two can hide structure that is essential for identification.

More precisely, I show that, whether or not integrable boundedness holds, the support-function approach remains sharp for a closely related object, which I call the \emph{moment closure} of the identified set. I then show that, under irreducibility, the identified set and its moment closure are indistinguishable in finite samples, in the sense that no size-controlled test can distinguish the null hypothesis that a parameter belongs to the identified set from the alternative that it belongs to the moment closure. Thus, once the model is written in irreducible form, the support-function approach captures the effective finite-sample limit of what can be learned from the data. Because a model can be reformulated to make such support implications explicit, these results provide a practical justification for using the support-function approach in empirically relevant counterfactual analyses that lie beyond the scope of existing sharp identification results.

The approach in this paper differs from the counterfactual predictive distribution set (CPDS) framework of \cite{kline_counterfactual_2024}, which conducts counterfactual analysis by first characterizing the identified set of the underlying structural model and then mapping that set, together with a specific parametric distribution of the latent variables, into a collection of distributions on the set of counterfactual outcomes. In contrast, I embed the counterfactual restrictions directly into an augmented support-and-moments model. This formulation accommodates moment restrictions on latent variables rather than requiring a fully specified latent-variable distribution, and it treats counterfactual parameters on the same footing as structural parameters. As a result, counterfactual parameters can be identified within the same framework, without a separate simulation or mapping step from a set-predicting model.

This perspective is closer in spirit to \cite{christensen_counterfactual_2023}, who also avoid simulating counterfactual outcomes and instead bound the counterfactual functional directly by optimizing over nonparametric neighborhoods of a benchmark latent-variable distribution. Their procedure, however, relies on additional regularity conditions tied to the choice of neighborhood, especially conditions governing integrability and the tail behavior of the latent distribution. Related computational contributions include \cite{gu_counterfactual_2025} and \cite{gu_dual_2023}. Those papers obtain powerful computational results in settings where the discrete structure of the model plays a central role: \cite{gu_counterfactual_2025} studies a general class of discrete outcome models and replaces the infinite-dimensional latent distribution with an identification-equivalent finite-dimensional representation, while \cite{gu_dual_2023} develops dual formulations for counterfactual bounds that accommodate Wasserstein neighborhoods.  By contrast, the framework in this paper is stated for general support-and-moments models and allows observed and latent variables with arbitrary support.

\subsubsection*{Notation}
Throughout the paper, capital letters (e.g., $X,Y,Z,U$) denote random variables or random vectors, and lower-case letters
(e.g., $x,y,z,u$) denote their realizations. All random variables are defined on a common complete probability space.
Vectors are written as column vectors, and superscript $\top$ denotes transpose. I use $\indicator\{\cdot\}$ for indicator functions. For example, $\indicator\{U\ge 0\}$ equals $1$ if $U\ge 0$ and equals $0$ otherwise. For vectors in finite-dimensional
Euclidean space, $\norm{\cdot}$ denotes the Euclidean norm. I write $\expt_F[\cdots]$ for expectation with respect to a
distribution $F$. When the underlying distribution is clear from context, I omit the subscript and write $\expt[\cdots]$. I
use $\equiv$ to denote definitions. Unless stated otherwise, all equalities and inequalities involving random variables are
understood to hold almost surely.

\subsubsection*{Overview}
The remainder of the paper is organized as follows. Section~\ref{sec:framework} introduces the support-and-moments
framework and presents motivating examples. Section~\ref{sec:counterfactual} formalizes counterfactual environments and
counterfactual parameters and shows how they can be embedded in an augmented model. Section~\ref{sec:identification}
develops the identification analysis, including the characterization of the moment closure and the finite-sample
indistinguishability results. Section~\ref{sec:conclusion} concludes. The Appendix collects proofs and additional
identification results.

\section{Theoretical Framework}\label{sec:framework}

\subsection{Structural model}
This paper studies a class of structural models characterized by a support restriction and a collection of moment restrictions:
\begin{equation}
  \prob\big[(U,Z)\in \Gamma(\theta)\big] = 1
  \qquad\text{and}\qquad
  \expt\big[r(U,Z;\theta)\big] = 0,
  \label{eq:framework}
\end{equation}
where $\theta\in\Theta$ denotes the (possibly vector-valued) parameter, $U$ is a vector of latent variables, and $Z=(X,Y)$ collects the observed variables, with $X$ a vector of covariates and $Y$ a vector of outcomes. Let $\mathcal{U}$, $\mathcal{X}$, and $\mathcal{Y}$ be Polish spaces containing the supports of $U$, $X$, and $Y$, respectively, and define $\mathcal{Z}\equiv \mathcal{X}\times \mathcal{Y}$. The parameter space $\Theta$ may be an arbitrary set. The set $\Gamma(\theta)\subseteq \mathcal{U}\times \mathcal{Z}$ imposes a $\theta$-dependent restriction on the joint support of $(U,Z)$.

Many structural models can be conveniently specified through a set-valued prediction rule of the form $Y \in \Gamma(X,U,\theta)$, where $\Gamma(X,U,\theta)\subseteq \mathcal{Y}$ is the set of outcomes consistent with $(X,U,\theta)$. In this case, the implied support set in \eqref{eq:framework} takes the form
\[
\Gamma(\theta)
\;=\;
\big\{(u,x,y)\in \mathcal{U}\times \mathcal{X}\times \mathcal{Y}:\; y\in \Gamma(x,u,\theta)\big\}.
\]
The structural model is \emph{incomplete} if, with positive probability, the correspondence $\Gamma(X,U,\theta)$ is not single-valued; that is, it contains multiple admissible outcomes.

For each $\theta\in\Theta$, the known function $r(\cdot,\cdot;\theta)$ maps $\mathcal{U}\times \mathcal{Z}$ to $\real^{d_r}$, where $d_r\equiv \dim(r)\ge 1$ is finite. As illustrated in the examples below, $r$ typically encodes moment restrictions on the latent variables (e.g., orthogonality or mean-independence conditions involving $U$ and the observed variables). Throughout the paper, I refer to the first condition in \eqref{eq:framework} as the \emph{support restriction} and to the second as the \emph{moment restriction}.

\begin{remark}\label{rmk:support_moment_difference}
Because a support restriction can be represented as a trivial moment restriction, such as $\expt\big[1-\indicator\{(U,Z)\in \Gamma(\theta)\}\big]=0$, some authors do not formally distinguish between support and moment restrictions (see, e.g., Example 1.4 in \cite{schennach_entropic_2014}). However, as discussed in Section \ref{sec:identification}, these two types of restrictions play fundamentally different roles in identification analysis. In particular, treating support restrictions as ordinary moment conditions can discard structure that is crucial for sharp identification, thereby yielding substantially weaker identification bounds.
\end{remark}

\subsection{Examples}
In what follows, I illustrate this framework with three examples.

\begin{example}[Static entry game with complete information] \label{ex:entry_game}
Consider a two-firm static entry game with complete information, as in
\cite{bresnahan_empirical_1991}, \cite{berry_estimation_1992}, and \cite{ciliberto_market_2009}.
In a given market, firm $j\in\{0,1\}$ observes a vector of profit shifters $X_j$ and a latent payoff shock $U_j$, and makes an entry decision $Y_j\in\{0,1\}$. Given an action profile $Y=(Y_0,Y_1)$, firm $j$'s payoff is
\[
\pi_j(Y,X_j,U_j)
\;=\;
Y_j\Big[X_j^\top \alpha - \Delta_j Y_{1-j} + U_j\Big],
\]
where the payoff from not entering ($Y_j=0$) is normalized to zero. The structural parameter is
$\theta=(\alpha,\Delta_0, \Delta_1)$. For notational convenience, write $X=(X_0,X_1)$ and $U=(U_0,U_1)$.

Assume $\Delta_j\ge 0$ for each $j\in\{0,1\}$. Under this restriction, for every $(X,U,\theta)$, the game admits at least one pure-strategy Nash equilibrium. The equilibrium conditions can be expressed as a support restriction by defining
\begin{equation}\label{eq:entry_game_support}
\Gamma(\theta)
\;=\;
\Big\{(u,x,y)\in \mathcal{U}\times \mathcal{X}\times \mathcal{Y}:\;
\forall j\in\{0,1\},\  
\pi_j(y,x_j,u_j)\ \ge\ \pi_j\big(\delta_j(y),x_j,u_j\big)\Big\},
\end{equation}
where $\delta_j(y)$ is the action profile obtained by alternating firm $j$'s action while holding its rival's action fixed. The condition $\prob\big[(U,Z)\in \Gamma(\theta)\big]=1$ with $Z\equiv(X,Y)$ then requires that the observed outcome $Y$ is almost surely consistent with a pure-strategy Nash equilibrium of the game.

Standard specifications (e.g., \cite{ciliberto_market_2009}) impose a parametric distribution on $U$ and typically assume
independence between $U$ and $X$. Within the present framework, the researcher can instead impose moment restrictions on $U$. Consider two illustrative cases:

\begin{itemize}
\item Suppose that for each firm $j$, $U_j$ has a conditional median of zero given $X$. Then,
	\begin{equation} \label{eq:entry_medium}
    \forall j\in\{0,1\},\qquad \expt\left[X\left(\indicator\{U_j\le 0\} - \frac{1}{2}\right)\right]=0.
\end{equation}

One may also impose additional symmetry restrictions at payoff-relevant thresholds. Specifically, let $c_j(x;\theta,y_{1-j})= -x_j^\top \alpha + \Delta_j y_{1-j}$ denote firm $j$'s indifference cutoff when its rival chooses $y_{1-j}\in\{0,1\}$, so that entering is a best response if and only if $U_j\ge c_j(X;\theta,y_{1-j})$. If the conditional distribution of $U_j$ given $X$ is symmetric about zero, then for any cutoff $c$ we have $\prob(U_j\ge c\mid X)=\prob(U_j\le -c\mid X)$. This yields the moment restrictions: $\forall j\in\{0,1\},\ \forall y_{1-j}\in\{0,1\}$,
\begin{equation}
  \expt\left[X\left(\indicator\{U_j\ge c_j(X;\theta,y_{1-j})\}-\indicator\{U_j\le -c_j(X;\theta,y_{1-j})\}\right)\right]=0.
  \label{eq:entry_symmetric}
\end{equation}

\item Suppose that for each firm $j$, $U_j$ has a mean of zero and is uncorrelated with the covariates $X$. Then,
	\begin{equation}\label{eq:entry_game_uncorr}
    \forall j\in \{0, 1\}, \qquad \expt[U_j] = 0, \quad \expt[U_j X] = 0.
\end{equation}
In addition, suppose the researcher has learned the variance of firm $j$'s revenue, denoted by $\sigma^2_{R, j}$. A plausible economic assumption is that the variance of the latent payoff shock $U_j$ is bounded above by the variance of firm $j$'s revenue. This implies
\begin{equation}\label{eq:entry_game_varaince}
    \forall j\in \{0, 1\}, \qquad \expt[U_j^2] = \tau_j \sigma^2_{R, j},
\end{equation}
where $\tau_j \in [0, 1]$ is an auxiliary parameter.
\end{itemize}
These moment restrictions are illustrative. Other restrictions on the joint distribution of $(U,X)$ can be incorporated
similarly. In each case, the model fits the general framework in \eqref{eq:framework}.
\end{example}

\begin{example}[Production-function estimation]
\label{ex:production_function}
Consider a production-function setting as in \cite{olley_dynamics_1996, levinsohn_estimating_2003, ackerberg_identification_2015}.
Let firms be indexed by $i$ and time by $t$. Suppose log output $Y_{it}$ depends on log capital $K_{it}$ and log variable
inputs $V_{it}$ according to
\[
Y_{it}
=
f(K_{it},V_{it};\beta) + \omega_{it} + \varepsilon_{it},
\]
where $\omega_{it}$ is Hicks-neutral productivity and $\varepsilon_{it}$ is a transitory shock (or measurement error in log
output). The parameter $\beta$ governs the production function. I assume that capital $K_{it}$ is predetermined as of period $t-1$, while $V_{it}$ is a flexible input freely chosen in period $t$.
Following the literature, I assume productivity evolves as a
first-order Markov process,
\[
\omega_{it}
=
g(\omega_{it-1};\gamma) + \xi_{it},
\]
where $g(\cdot;\gamma)$ is a known functional form indexed by $\gamma$, and $\xi_{it}$ is the productivity innovation.

Let $W_{it}$ be a vector of instruments satisfying the exogeneity conditions
$\expt[\xi_{it}\mid W_{it}]=0$, $\expt[\varepsilon_{it}\mid W_{it}]=0$, and $\expt[\varepsilon_{it-1}\mid W_{it}]=0$.
These assumptions imply moment restrictions involving the observed variables
$(Y_{it},K_{it},V_{it},Y_{it-1},K_{it-1},V_{it-1},W_{it})$ and the latent productivities $(\omega_{it},\omega_{it-1})$:
\begin{equation}
\label{eq:prod_function_moment}
\begin{aligned}
\expt\Big[W_{it}\big(Y_{it} - f(K_{it},V_{it};\beta) - \omega_{it}\big)\Big] &= 0,\\
\expt\Big[W_{it}\big(Y_{it-1} - f(K_{it-1},V_{it-1};\beta) - \omega_{it-1}\big)\Big] &= 0,\\
\expt\Big[W_{it}\big(\omega_{it} - g(\omega_{it-1};\gamma)\big)\Big] &= 0.
\end{aligned}
\end{equation}

A maintained assumption in this literature is the \emph{invertibility} condition: there exist observed covariates
$X_{it}$ and a function $h$ such that $\omega_{it}=h(X_{it})$. This assumption can be restrictive regarding demand and markup heterogeneity; see, for example, the discussion in \cite{doraszelski_production_2025}. Here, instead of imposing invertibility, I use an inequality restriction implied by cost
minimization.

Specifically, suppose firms minimize costs when choosing $V_{it}$. Then the markup $\mu_{it}$ and productivity $\omega_{it}$
satisfy
\[
\log(\mu_{it})
=
P_{it} + f(K_{it},V_{it};\beta) + \omega_{it}
- P^V_{it} - V_{it}
+ \log\big(f_V(K_{it},V_{it};\beta)\big),
\]
where $P_{it}$ and $P^V_{it}$ denote, respectively, the log output price and the log price of variable inputs, and $f_V$ is
the partial derivative of $f$ with respect to $V$. If markups satisfy $\mu_{it}\ge 1$ almost surely (that is, price weakly
exceeds marginal cost), then $\log(\mu_{it})\ge 0$, and hence
\begin{equation}
\label{eq:markup_constraint}
\omega_{it}
\;\ge\;
\ell(P^V_{it},V_{it},P_{it},K_{it};\beta)
\;\equiv\;
P^V_{it} + V_{it} - P_{it}
- f(K_{it},V_{it};\beta)
- \log\big(f_V(K_{it},V_{it};\beta)\big)
\end{equation}
holds almost surely.

To map this example into the framework in \eqref{eq:framework}, collect the observed variables appearing in
\eqref{eq:prod_function_moment} and \eqref{eq:markup_constraint} (and their lagged counterparts) into a vector
$Z_{it}$, and define the latent vector
$U_{it}\equiv(\omega_{it},\omega_{it-1})$.\footnote{%
For expositional convenience, the support restriction below uses two consecutive periods $(t-1,t)$. Extensions to longer
panels are straightforward.
}
Then \eqref{eq:markup_constraint} and its lagged counterpart at $t-1$ induce the support restriction
$\prob\big[(U_{it},Z_{it})\in\Gamma(\theta)\big]=1$, with $\theta=(\beta,\gamma)$ and
\[
\Gamma(\theta)
=
\Big\{(u,z):\ u=(\omega_{it},\omega_{it-1}),\
\omega_{it} \ge \ell(P^V_{it},V_{it},P_{it},K_{it};\beta),\ 
\omega_{it-1} \ge \ell(P^V_{it-1},V_{it-1},P_{it-1},K_{it-1};\beta)\Big\}.
\]
Together with the moment restrictions in \eqref{eq:prod_function_moment}, this example fits the general framework in
\eqref{eq:framework}.
\end{example}

The final example in this section is a linear regression model with an interval-censored dependent variable. I use this
simple setting as a heuristic example in Section \ref{sec:identification} to highlight several technical subtleties in the identification analysis, including the importance of distinguishing support restrictions from moment restrictions, as emphasized in
Remark~\ref{rmk:support_moment_difference}.

\begin{example}[Regression with an interval-censored dependent variable]
\label{ex:interval_censored_regression}
Consider a scalar linear regression model $Y^{*} = \alpha + \beta W + \varepsilon$,
where $Y^{*}$ is a latent outcome, $W$ is an observed regressor, and $\theta=(\alpha,\beta)\in\Theta\subseteq\real^{2}$ is the parameter of interest. The standard exogeneity condition $\expt[\varepsilon\mid W]=0$ implies the unconditional orthogonality conditions
\begin{equation}
\label{eq:interval_reg_moments}
\begin{aligned}
\expt\big[Y^{*}-\alpha-\beta W\big] &= 0,\\
\expt\Big[W\big(Y^{*}-\alpha-\beta W\big)\Big] &= 0.
\end{aligned}
\end{equation}

Suppose $Y^{*}$ is not observed directly but is known to lie within an observed interval
$[\underline{Y},\overline{Y}]$:
\begin{equation}\label{eq:example_interval_support_prob_1}
    \prob\big[Y^* \in [\underline{Y}, \overline{Y}]\big] = 1.
\end{equation}
Let $Z\equiv(\underline{Y},\overline{Y},W)$ collect the observed variables and define the latent variable $U\equiv Y^{*}$. The
interval constraint in \eqref{eq:example_interval_support_prob_1} induces the support restriction $\prob\big[(U,Z)\in\Gamma(\theta)\big]=1$, where
\begin{equation}
\label{eq:interval_regression_support}
\Gamma(\theta)\equiv
\big\{(u,z): z=(\underline{y},\overline{y},w)\ \text{and}\ \underline{y}\le u\le \overline{y}\big\}.
\end{equation}
The moment restrictions in \eqref{eq:interval_reg_moments} can be written as $\expt\big[r(U,Z;\theta)\big]=0$, where
\begin{equation}
\label{eq:interval_regression_moment}
r(u,z;\theta)
=
\big(u-\alpha-\beta w,\ w(u-\alpha-\beta w)\big)^\top,
\qquad z=(\underline{y},\overline{y},w).
\end{equation}
Hence, this model fits the general framework in \eqref{eq:framework}.

This setup extends straightforwardly to multiple regressors and to settings in which some regressors are also interval-censored. I focus on the scalar case here to keep the notation simple and to highlight the main identification issues as transparently as possible.
\end{example}

\section{Counterfactual Analysis}\label{sec:counterfactual}
Building on the structural model in \eqref{eq:framework}, this section formalizes the counterfactual analysis. Any counterfactual exercise requires two ingredients. First, the researcher must specify how the counterfactual outcome $\tilde{Y}$ is generated under the counterfactual environment. Specifically, the analyst must define how $\tilde{Y}$ relates to the baseline observed and latent variables $(Z, U)$, the structural parameter $\theta$, and, potentially, additional latent primitives $\tilde{U}$. Second, one must specify the counterfactual parameter(s) of interest, namely, the specific functional(s) of the counterfactual distribution to be evaluated. I discuss these two components in turn.

\subsection{Counterfactual outcomes}
To describe the counterfactual environment, I assume there exists a correspondence $\tilde{\Gamma}(\cdot,\cdot;\theta)$ such that, in the counterfactual setting,
\begin{equation}
\prob\big[(\tilde{Y},\tilde{U}) \in \tilde{\Gamma}(Z,U;\theta)\big] = 1. \label{eq:counterfactual_support}
\end{equation}
The correspondence $\tilde{\Gamma}$ typically combines (\emph{i}) the baseline structural restrictions with (\emph{ii}) the counterfactual intervention (e.g., a policy change, a change in primitives, or a change in market structure). In many counterfactual scenarios, $\tilde{U}$ is redundant. For instance, it may equal $U$ or be a deterministic function of $U$. Alternatively, the counterfactual intervention may introduce new latent primitives $\tilde{U}$ that are not uniquely determined by $(Z,U)$. In such cases, the counterfactual analysis requires additional restrictions on the joint distribution of the counterfactual variables. To accommodate this, I allow for supplementary moment restrictions of the form
\begin{equation}
\expt\big[\tilde{r}(\tilde{Y},\tilde{U}, Z,U;\theta)\big] = 0,
\label{eq:counterfactual_moment_latent}
\end{equation}
where $\tilde{r}$ is a known vector-valued function. As detailed below, \eqref{eq:counterfactual_moment_latent} would be combined with the baseline moment restriction in \eqref{eq:framework} to identify counterfactual objects.

I first illustrate the construction of $\tilde{\Gamma}$ and $\tilde{r}$ using Examples~\ref{ex:entry_game} and \ref{ex:production_function}. I then discuss the limitations of this approach and how it relates to alternative formulations of counterfactuals.

\begin{example}[continues=ex:entry_game]
Consider the entry game in Example~\ref{ex:entry_game}. Recall that $Z\equiv(X, Y)$ consists of the observed profit shifters $X$ and entry decisions $Y$. To illustrate how to construct $\tilde{\Gamma}$, consider the following counterfactual interventions.

\smallskip
\noindent
\emph{(i) A change in profit shifters.}
Suppose the counterfactual transforms profit shifters from $X$ to $\tilde{X}=\phi(X)$ for a known mapping $\phi:\mathcal{X}\to\mathcal{X}$, while leaving payoff shocks unchanged, that is, $\tilde{U}=U$. For instance, $\phi$ may capture a hypothetical change in market size or demand. If the counterfactual entry profile $\tilde{Y}$ is assumed to be a pure-strategy Nash equilibrium of the modified game, then $\prob\big[(\tilde{Y},\tilde{U})\in \tilde{\Gamma}(Z,U;\theta)\big]=1$, where\footnote{This correspondence $\tilde{\Gamma}$ depends on $z\equiv(x,y)$ only through $x$. If the researcher is willing to restrict the counterfactual equilibrium selection rule so that it preserves the status quo equilibrium whenever it remains an equilibrium under the counterfactual, then $\tilde{\Gamma}$ may also depend on $y$.}
\begin{multline}\label{eq:entry_counterfactual_support}
\tilde{\Gamma}(z,u;\theta)
\;\equiv\;
\Big\{(\tilde{y},\tilde{u})\in \{0,1\}^2\times \real^2:\ 
\tilde{u}=u,\ \text{and } \\ \forall j\in\{0,1\},\ 
\pi_j\big(\tilde{y},\phi_j(x),u_j\big)\ \ge\ \pi_j\big(\delta_j(\tilde{y}),\phi_j(x),u_j\big)\Big\}.
\end{multline}
Here $\phi_j(x)$ denotes the counterfactual covariates relevant for firm $j$, and $\delta_j(\tilde{y})$ is the action profile obtained by unilaterally changing firm $j$'s action while holding its rival's action fixed. In this intervention, $\tilde{U}$ is determined by $U$, so additional restrictions such as \eqref{eq:counterfactual_moment_latent} are redundant.

\smallskip
\noindent
\emph{(ii) A merger.}
As a second counterfactual case, suppose the two potential entrants merge into a single firm that makes a binary entry decision $\tilde{Y}\in\{0,1\}$. Suppose the merged firm inherits profit shifters and shocks according to $\tilde{X}=\max\{X_0,X_1\}$ and $\tilde{U}=\max\{U_0,U_1\}$, where the maximum for $\tilde{X}$ is taken componentwise.\footnote{%
Alternative merger rules can be accommodated. For example, one could impose $\tilde{U}\in[\min\{U_0,U_1\},\max\{U_0,U_1\}]$ and $\tilde{X}=\lambda X_0+(1-\lambda)X_1$ for a latent $\lambda\in[0,1]$.%
}
Assume the merged firm's net profit from entry is $\tilde{X}^\top \alpha+\tilde{U}$. Then $\prob\big[(\tilde{Y},\tilde{U})\in \tilde{\Gamma}(Z,U;\theta)\big]=1$, where 
\[
\tilde{\Gamma}(z,u;\theta)
\;\equiv\;
\Big\{(\tilde{y},\tilde{u})\in \{0,1\}\times \real:\ 
\tilde{u}=\max\{u_0,u_1\},\ 
(-1)^{\tilde{y}}\Big[\max\{x_0,x_1\}^\top \alpha+\tilde{u}\Big]\le 0\Big\}.
\]
This correspondence is single-valued except in the knife-edge event where the merged firm's entry profit equals zero, in which case both actions are consistent with optimality. As in case (i), $\tilde{U}$ is determined by $U$, so \eqref{eq:counterfactual_moment_latent} is redundant.

\smallskip
\noindent
\emph{(iii) A new competitor.}
As a third counterfactual, suppose a third potential entrant $j=2$ becomes eligible to enter the market. Assume that the researcher can construct the entrant's observed profit shifters as $X_2=\varphi(X_0,X_1)$ and its competitive effect as $\Delta_2=\phi(\Delta_0,\Delta_1)$ for known mappings $\varphi$ and $\phi$. Let $\tilde{X}\equiv (X_0,X_1,X_2)$. Assume the incumbents' payoff shocks remain the same in the counterfactual, so $\tilde{U}_j=U_j$ for $j\in\{0,1\}$, but allow the new entrant's shock $\tilde{U}_2$ to be unrestricted.

Suppose the payoff structure extends to three firms as
\[
\pi_j(y,\tilde{X}_j,\tilde{U}_j)
\;=\;
y_j\Big[\tilde{X}_j^\top \alpha - \Delta_j \sum_{k\neq j} y_k + \tilde{U}_j\Big],
\qquad j\in\{0,1,2\}.
\]
If the counterfactual outcome $\tilde{Y}\in\{0,1\}^3$ is assumed to be a pure-strategy Nash equilibrium, then $\prob\big[(\tilde{Y},\tilde{U})\in \tilde{\Gamma}(Z,U;\theta)\big]=1$, where
\begin{multline*}
\tilde{\Gamma}(z,u;\theta)
\;\equiv\;
\Big\{(\tilde{y},\tilde{u})\in \{0,1\}^3\times \real^3:\ 
\tilde{u}_0=u_0,\ \tilde{u}_1=u_1,\ \\
\text{and } \forall j\in\{0,1,2\},\
\pi_j\big(\tilde{y},\tilde{x}_j,\tilde{u}_j\big)\ \ge\ \pi_j\big(\delta_j(\tilde{y}),\tilde{x}_j,\tilde{u}_j\big)\Big\},
\end{multline*}
with $\tilde{x}_0 = x_0$, $\tilde{x}_1 = x_1$, $\tilde{x}_2 = \varphi(x_0, x_1)$, and $\delta_j(\tilde{y})$ defined as the action profile obtained by unilaterally changing firm $j$'s action while holding the other firms' actions fixed. 

Unlike cases (i) and (ii), this counterfactual introduces a new latent primitive $\tilde{U}_2$ that is not uniquely determined by $(X,U)$. It is therefore natural to supplement the equilibrium restriction with additional assumptions on $\tilde{U}_2$, for example, by imposing the same median restriction used for the incumbents. One convenient formulation is the moment condition \eqref{eq:counterfactual_moment_latent} with
\[
\tilde{r}(\tilde{y},\tilde{u},z,u;\theta)
=
x\left(\indicator\{\tilde{u}_2\le 0\} - \frac{1}{2}\right).
\]
Moment restrictions like \eqref{eq:entry_symmetric}-\eqref{eq:entry_game_varaince} can be imposed to $\tilde{U}_2$ as well.
\end{example}

\begin{example}[continues=ex:production_function]
Consider the production-function context in Example~\ref{ex:production_function}. Suppose the parameter of interest is the short-run supply curve in period $t$. To trace out supply at different prices, consider a counterfactual in which the log output price is set to $\tilde{P}$, while the state variables $K_{it}$ and $\omega_{it}$, as well as the log variable-input price $P^V_{it}$, are held fixed. The firm then chooses the log variable input to maximize period-$t$ profit in levels. Let the counterfactual input demand correspondence be
\[
\tilde{V}_{it}
\in
\varphi(\tilde{P},K_{it},P^V_{it},\omega_{it};\beta)
\;\equiv\;
\argmax_{v}
\Big\{
\exp\big(\tilde{P}+f(K_{it},v;\beta)+\omega_{it}\big)
-
\exp\big(P^V_{it}+v\big)
\Big\},
\]
where the $\argmax$ operator is set-valued to accommodate multiple maximizers.

Let $\tilde{Y}_{it}$ denote counterfactual log output. Under the production technology, $\tilde{Y}_{it} = f(K_{it},\tilde{V}_{it};\beta)+\omega_{it}$. Equivalently, the counterfactual outcome can be summarized by the correspondence
\[
\tilde{Y}_{it}
\in
\tilde{\Gamma}(\tilde{P},K_{it},P^V_{it},\omega_{it};\beta)
\;\equiv\;
\Big\{f(K_{it},v;\beta)+\omega_{it}:\ v\in \varphi(\tilde{P},K_{it},P^V_{it},\omega_{it};\beta)\Big\}.
\]
Because the counterfactual holds productivity fixed, no additional latent primitive is introduced, making $\tilde{U}$ redundant in this example.
\end{example}

\subsection{Counterfactual parameters} \label{subsec:counterfactual_parameters} 
Given a counterfactual environment specified by the support restriction \eqref{eq:counterfactual_support} and, when needed, the additional moment restriction \eqref{eq:counterfactual_moment_latent}, the next step is to define the counterfactual parameter of interest. Let $\tilde{\theta}\in\tilde{\Theta}$ denote a (possibly vector-valued) counterfactual parameter. I assume that $\tilde{\theta}$ is defined as a solution to a moment condition of the form
\begin{equation}
\label{eq:counterfactual_parameter_def}
\expt\Big[\tilde{m}(\tilde{Y},\tilde{U}, Z,U;\theta,\tilde{\theta})\Big] = 0,
\end{equation}
where the known function $\tilde{m}$ encodes the definition of the counterfactual object. This formulation is flexible enough to cover many targets routinely reported in applications, including means, probabilities, policy effects, and distributional functionals such as quantiles (under standard regularity conditions). I illustrate the construction in Examples~\ref{ex:entry_game} and \ref{ex:production_function}.

\begin{example}[continues=ex:entry_game]
Recall that $\tilde{Y}$ denotes the counterfactual entry profile, $\tilde{X}$ denotes counterfactual profit shifters (e.g., $\tilde{X}=\phi(X)$ for a known mapping $\phi$), and $\tilde{U}$ denotes counterfactual payoff shocks (as determined by the specific counterfactual intervention). Let $\tilde{\pi}_j$ denote firm $j$'s counterfactual profit. For instance, in the baseline two-firm environment,
\[
\tilde{\pi}_j
\;\equiv\;
\tilde{Y}_j\Big[\tilde{X}_j^\top \alpha - \Delta_j \tilde{Y}_{1-j} + \tilde{U}_j\Big],
\qquad j\in\{0,1\},
\]
with the natural modification for counterfactuals that alter the market structure (e.g., a merger or the entry of an additional firm).

Examples of counterfactual parameters $\tilde{\theta}$ and the corresponding functions $\tilde{m}(\cdot)$ that fit \eqref{eq:counterfactual_parameter_def} include:
\begin{enumerate}
\item \emph{Expected number of entrants:}
\[
\tilde{\theta} \equiv \expt\Big[\sum_{j}\tilde{Y}_j\Big],
\qquad\text{so that}\qquad
\tilde{m} = \sum_{j}\tilde{Y}_j - \tilde{\theta}.
\]

\item \emph{Probability that the market is not served:}
\[
\tilde{\theta} \equiv \prob\Big[\sum_{j}\tilde{Y}_j=0\Big],
\qquad\text{so that}\qquad
\tilde{m} = \indicator\Big\{\sum_{j}\tilde{Y}_j=0\Big\}-\tilde{\theta}.
\]

\item \emph{Expected total firm surplus (total profit):}
\[
\tilde{\theta} \equiv \expt\Big[\sum_{j}\tilde{\pi}_j\Big],
\qquad\text{so that}\qquad
\tilde{m} = \sum_{j}\tilde{\pi}_j - \tilde{\theta}.
\]

\item \emph{Average profit conditional on entry for firm $j$:}
Provided $\prob\big[\tilde{Y}_j=1\big]>0$, define
\[
\tilde{\theta}_j \equiv \expt\big[\tilde{\pi}_j \mid \tilde{Y}_j=1\big].
\]
Because $\tilde{\pi}_j=0$ when $\tilde{Y}_j=0$, this conditional mean can be characterized by the unconditional moment restriction
\[
\expt\big[\tilde{\pi}_j - \tilde{\theta}_j \tilde{Y}_j\big]=0,
\qquad\text{that is,}\qquad
\tilde{m} = \tilde{\pi}_j - \tilde{\theta}_j \tilde{Y}_j.
\]
\end{enumerate}
\end{example}

\begin{example}[continues=ex:production_function]
In the production-function counterfactual described above, let $\tilde{P}$ denote the counterfactual log output price, and let $(\tilde{V}_{it},\tilde{Y}_{it})$ denote the resulting counterfactual input choice and log output. Define period-$t$ counterfactual profit in levels as
\[
\tilde{\pi}_{it}
\;\equiv\;
\exp\big(\tilde{P}+\tilde{Y}_{it}\big) - \exp\big(P^V_{it}+\tilde{V}_{it}\big),
\]
where $P^V_{it}$ is the observed log price of variable inputs.

Examples of counterfactual parameters include:
\begin{enumerate}
\item \emph{Average (or aggregate) output in levels:}
\[
\tilde{\theta} \equiv \expt\big[\exp(\tilde{Y}_{it})\big],
\qquad\text{so that}\qquad
\tilde{m} = \exp(\tilde{Y}_{it}) - \tilde{\theta}.
\]

\item \emph{Average profit:}
\[
\tilde{\theta} \equiv \expt\big[\tilde{\pi}_{it}\big],
\qquad\text{so that}\qquad
\tilde{m} = \tilde{\pi}_{it} - \tilde{\theta}.
\]

\item \emph{$\tau$-quantile of the profit distribution:}
Let $\tilde{\theta}$ denote the $\tau$-quantile of $\tilde{\pi}_{it}$. Under standard continuity conditions at the quantile, $\tilde{\theta}$ satisfies
\[
\expt\big[\indicator\{\tilde{\pi}_{it}\le \tilde{\theta}\}-\tau\big]=0,
\qquad\text{that is,}\qquad
\tilde{m}=\indicator\{\tilde{\pi}_{it}\le \tilde{\theta}\}-\tau.
\]
\end{enumerate}
\end{example}

\subsection{Unified identification framework for structural and counterfactual parameters}
As the preceding discussion makes clear, counterfactual analysis typically adds three types of restrictions to the baseline model in \eqref{eq:framework}: (\emph{i}) a counterfactual support restriction \eqref{eq:counterfactual_support}; (\emph{ii}) an auxiliary moment restriction \eqref{eq:counterfactual_moment_latent}, when the counterfactual analysis introduces additional latent primitives; and (\emph{iii}) a defining moment restriction for the counterfactual parameter $\tilde{\theta}$, given in \eqref{eq:counterfactual_parameter_def}.

A key observation is that the combined system of baseline and counterfactual restrictions can itself be formulated as a model of the form in \eqref{eq:framework}. To see this, define an augmented latent vector $ U' \equiv (U,\tilde{Y},\tilde{U})$, which collects the baseline latent variables together with the unobserved counterfactual outcomes and counterfactual latent primitives. Let the observed vector remain $Z=(X,Y)$, and define the augmented parameter vector $\theta' \equiv (\theta,\tilde{\theta})$. Then, the baseline support restriction $\prob\big[(U,Z)\in\Gamma(\theta)\big]=1$ and the counterfactual support restriction $\prob\big[(\tilde{Y},\tilde{U})\in\tilde{\Gamma}(Z,U;\theta)\big]=1$ can be summarized by a single joint support restriction:
\begin{equation}
\label{eq:combined_support}
\prob\big[(U',Z)\in \Gamma'(\theta')\big]=1
\ \text{ with }\ 
\Gamma'(\theta')
\equiv
\Big\{(u',z):\  (u,z)\in \Gamma(\theta),\ (\tilde{y},\tilde{u})\in \tilde{\Gamma}(z,u;\theta)\text{ and }u'=(u,\tilde{y},\tilde{u})\Big\}.
\end{equation}

Likewise, the baseline moment restriction, the auxiliary counterfactual moment restriction (when needed), and the defining moment restriction for $\tilde{\theta}$ can be stacked into a single augmented moment restriction:
\begin{equation}
\label{eq:combined_moment}
\expt\big[r'(U',Z;\theta')\big]=0,
\qquad
r'(u',z;\theta')
\equiv
\begin{pmatrix}
r(u,z;\theta)\\
\tilde{r}(\tilde{y},\tilde{u},x,y,u;\theta)\\
\tilde{m}(\tilde{y},\tilde{u},x,y,u;\theta,\tilde{\theta})
\end{pmatrix}.
\end{equation}
When no additional counterfactual moment restrictions are required, the block $\tilde{r}(\cdot)$ can simply be omitted.

Equations \eqref{eq:combined_support} and \eqref{eq:combined_moment} demonstrate that counterfactual analysis can be embedded within the same support-and-moments framework as the baseline structural model. This yields an important implication: identification and inference for the counterfactual parameter $\tilde{\theta}$ can be carried out using the same partial-identification techniques designed for the structural parameter $\theta$. In particular, the identified set for $\tilde{\theta}$ can be obtained as the projection of the identified set for $\theta'$ in the augmented model. When only $\tilde{\theta}$ is of interest, inference can proceed by applying standard subvector inference procedures within the augmented model, leveraging existing methods for partially identified models; see, for example, \cite{bugni_inference_2017, kaido_confidence_2019, belloni_subvector_2018, marcoux_simple_2024}.

This approach differs fundamentally from the traditional two-step workflow, in which the researcher first identifies or estimates $\theta$ and then computes $\tilde{\theta}$ by simulating counterfactual outcomes from the structural model. In incomplete models, simulation typically requires additional equilibrium-selection or model-completion assumptions, and it may be entirely infeasible when the model delivers set-valued predictions. By contrast, the augmented-model formulation treats $\tilde{\theta}$ on an equal footing with $\theta$ for identification purposes, thereby avoiding the need to simulate counterfactual outcomes from a set-predicting structural model.

\section{Identification Analysis}\label{sec:identification}
This section studies identification within the support-and-moments framework of \eqref{eq:framework}. As demonstrated in the previous section, the same analysis applies not only to the baseline structural parameter $\theta$, but also to counterfactual parameters $\tilde{\theta}$ once the counterfactual restrictions are embedded in the augmented model defined by \eqref{eq:combined_support} and \eqref{eq:combined_moment}.

The remainder of this section proceeds as follows. First, I review sharp identification results from the existing literature and explain why they are often not directly applicable to the present setting. In particular, the key regularity condition used for sharp identification in the literature, namely the integrable boundedness of the relevant random sets, frequently fails in empirically relevant counterfactual exercises. Second, I characterize what can be learned when integrable boundedness is violated. This analysis yields a set that I call the \emph{moment closure} of the identified set. Finally, I study the relationship between the identified set and its moment closure, showing that, under weak regularity conditions, the two sets are statistically indistinguishable in finite samples. Importantly, these regularity conditions depend on how the model is formulated and are closely related to the distinction between support and moment restrictions emphasized in Remark~\ref{rmk:support_moment_difference}.

\subsection{Identified set and sharp identification results}
I begin by defining the identified set. Let $\mathcal{F}$ denote a collection of candidate distributions for $Z$. For example, $\mathcal{F}$ may be the class of all Borel probability measures on $\mathcal{Z}$, or a smaller class reflecting additional assumptions maintained by the researcher.  Fix $\theta\in\Theta$ and $F\in\mathcal{F}$. Let $\mathcal{H}(\theta,F)$ denote the set of all joint distributions $H$ of $(U,Z)$ such that (i) the support restriction holds $H$-almost surely and (ii) the marginal distribution of $Z$ under $H$ equals $F$:
\[
\mathcal{H}(\theta,F)
\equiv
\Big\{
H:\ \prob_H\big[(U,Z)\in \Gamma(\theta)\big]=1
\ \text{and}\ 
H_Z = F
\Big\},
\]
where $H_Z$ denotes the $Z$-marginal distribution of $H$.

\begin{definition}[Identified set]\label{def:id_set}
For any $F\in\mathcal{F}$, the identified set, denoted by $\Theta_I(F)$, is the set of all parameter values $\theta\in\Theta$ for which there exists some $H\in\mathcal{H}(\theta,F)$ satisfying the moment restriction $\expt_H\big[r(U,Z;\theta)\big]=0$. That is,
\[
\Theta_I(F)
\equiv
\Big\{
\theta\in\Theta:\ \exists\, H\in\mathcal{H}(\theta,F)\ \text{such that}\ \expt_H\big[r(U,Z;\theta)\big]=0
\Big\}.
\]
Equivalently, for any norm $\norm{\cdot}$ on $\real^{d_r}$, $\theta\in\Theta_I(F)$ if
\begin{equation}
    \label{eq:id_def_min}
\min_{H\in\mathcal{H}(\theta,F)}\norm{\expt_H\big[r(U,Z;\theta)\big]} = 0,
\end{equation}
with the convention that the minimum over an empty set equals $+\infty$.
\end{definition}

Sharp identification results for $\Theta_I(F)$ have been established by \cite{ekeland_optimal_2010} and \cite{beresteanu_sharp_2011}, who develop what is now commonly referred to as the \emph{support-function approach}. I adapt their approach to the present support-and-moments framework.

Fix $\theta\in\Theta$ and $z\in\mathcal{Z}$. Define
\begin{equation}\label{eq:moment_random_set}
\mathcal{R}(z;\theta)
\;\equiv\;
\big\{\, r(u,z;\theta)\in \real^{d_r}:\ (u,z)\in \Gamma(\theta)\,\big\}
\end{equation}
as the set of all values the moment function can take at $z$ under $\theta$ subject to the support restriction. Recall that $d_r\equiv \dim(r)\ge 1$, and let $\mathcal{S}\equiv\{\lambda\in\real^{d_r}:\norm{\lambda} =  1\}$ denote the unit sphere in $\real^{d_r}$. For any $\lambda\in\mathcal{S}$ and $z\in\mathcal{Z}$, define the support function of $\mathcal{R}(z;\theta)$ in direction $\lambda$ by
\begin{equation}\label{eq:support_function_def}
\gamma(\lambda,z;\theta)
\;\equiv\;
\sup_{t\in \mathcal{R}(z;\theta)} \lambda^\top t.
\end{equation}
Given $F\in\mathcal{F}$, the support-function approach relies on the following moment inequality:
\begin{equation}
\label{eq:moment_ineq}
\inf_{\lambda\in\mathcal{S}} \expt_F\big[\gamma(\lambda,Z;\theta)\big] \ge 0.
\end{equation}
Inequality \eqref{eq:moment_ineq} provides a sharp characterization of the identified set $\Theta_I(F)$ under certain regularity conditions. I organize these conditions into two groups.

\begin{assumption}\label{assu:reg}
Assume $d_r\ge 1$. For any $\theta\in\Theta$ and any $F\in\mathcal{F}$, assume the following:
\begin{enumerate}
\item\label{enu:CC1}
$\Gamma(\theta)$ is a Borel subset of $\mathcal{U}\times\mathcal{Z}$, and the section $\Gamma(z;\theta)\equiv \{u\in\mathcal{U}:(u,z)\in\Gamma(\theta)\}$ is a nonempty Borel set for every $z\in\mathcal{Z}$. Moreover, $r(u,z;\theta)$ is Borel measurable on $\mathcal{U}\times\mathcal{Z}$.

\item\label{enu:CC2}
There exists a Borel measurable function $g(\cdot;\theta,F)$ such that $\expt_F\big[g(Z;\theta,F)\big]<\infty$ and, for $F$-almost every $z$,
\[
g(z;\theta,F)\ \ge\ \inf\big\{\norm{t}:\ t\in \mathcal{R}(z;\theta)\big\}.
\]
\end{enumerate}
\end{assumption}

Assumption~\ref{assu:reg} collects basic measurability and well-posedness conditions. Condition~\ref{enu:CC1} is standard. Condition~\ref{enu:CC2} ensures that, for a given $(\theta,F)$, moment $\expt_H\big[r(U,Z;\theta)\big]$ are well-defined for at least one admissible joint distribution $H$ in $\mathcal{H}(\theta, F)$.

\begin{assumption}\label{assu:compact}
For any $\theta\in\Theta$ and any $F\in\mathcal{F}$, assume the following:
\begin{enumerate}
\item\label{enu:CC3}
For $F$-almost every $z$, the set $\mathcal{R}(z;\theta)$ is closed.

\item\label{enu:CC4}
There exists a Borel measurable function $g(\cdot;\theta)$ such that $\expt_F\big[g(Z;\theta)\big]<\infty$ and, for $F$-almost every $z$,
\[
g(z;\theta)\ \ge\ \sup\big\{\norm{t}:\ t\in \mathcal{R}(z;\theta)\big\}.
\]
\end{enumerate}
\end{assumption}

Assumption~\ref{assu:compact} corresponds to the closedness and integrable boundedness conditions routinely invoked in the support-function literature. In particular, Assumption~\ref{assu:compact} implies that $\mathcal{R}(z;\theta)$ is compact for $F$-almost every $z$. As discussed below, this requirement can be restrictive in structural and counterfactual applications because $\mathcal{R}(z;\theta)$ may fail to be bounded.

\begin{theorem}
\label{thm:sharp_support_function}
Under Assumptions~\ref{assu:reg} and \ref{assu:compact}, for any $F\in\mathcal{F}$ and any $\theta\in\Theta$, $\theta\in\Theta_I(F)$ if and only if $\theta$ satisfies \eqref{eq:moment_ineq}. Equivalently, for any $F\in\mathcal{F}$,
\begin{equation*}
\Theta_I(F)=\tilde{\Theta}(F),
\qquad
\text{where }\;
\tilde{\Theta}(F)\equiv
\left\{
\theta\in\Theta:\ \inf_{\lambda\in\mathcal{S}} \expt_F\big[\gamma(\lambda,Z;\theta)\big]\ge 0
\right\}.
\end{equation*}
\end{theorem}

Here, $\tilde{\Theta}(F)$ denotes the set of parameter values satisfying the support-function inequality in \eqref{eq:moment_ineq}. In what follows, I refer to $\tilde{\Theta}(F)$ as the \emph{support-function set}. Although Theorem~\ref{thm:sharp_support_function} shows that the identified set $\Theta_I(F)$ coincides with the support-function set $\tilde{\Theta}(F)$, this sharp characterization relies on Assumption~\ref{assu:compact}. As the examples below illustrate, that assumption can fail in economically relevant structural models and counterfactual applications.

One common source of failure is that the support restriction leaves some components of the latent variables unbounded conditional on observables, while the moment vector---whether from the baseline structural model or from the counterfactual analysis augmentation---depends on those latent variables in levels or higher-order terms. In such cases, the random set $\mathcal{R}(z;\theta)$ need not be integrably bounded. A second source of failure is representational: Assumption~\ref{assu:compact} may also fail when a support restriction is implicitly implied from a moment restriction rather than imposed explicitly. Importantly, neither type of failure implies, by itself, that the model lacks identifying power. Indeed, the examples show that all three different scenarios can arise when Assumption~\ref{assu:compact} is violated: (\emph{i}) both the support-function set $\tilde{\Theta}(F)$ and the identified set $\Theta_I(F)$ may be unbounded; (\emph{ii}) both sets may continue to deliver finite bounds; and (\emph{iii}) the support-function set $\tilde{\Theta}(F)$ may be completely uninformative (that is, equal to the entire parameter space) even though the identified set $\Theta_I(F)$ remains informative.

These possibilities raise fundamental questions that Theorem~\ref{thm:sharp_support_function} alone cannot answer. Suppose Assumption~\ref{assu:compact} fails. If the support-function set $\tilde{\Theta}(F)$ successfully provides finite bounds, does it still sharply characterize the identified set? Conversely, if $\tilde{\Theta}(F)$ is unbounded or completely uninformative, does this reflect a methodological failure of the support-function approach, or does it instead indicate that the maintained model itself lacks the identifying power to bound the parameter of interest? Distinguishing between these alternative explanations is essential in empirical practice: a methodological failure suggests the need for an alternative identification approach, whereas an inherently uninformative model necessitates stronger economic assumptions. Subsections~\ref{sec:moment_closure} and \ref{subsec:finite_sample_indistinguishability} develop formal results that resolve these questions. Before proceeding to the formal theory, however, I first present concrete examples illustrating the three scenarios. \vspace{0.5em}

\noindent {\bf Illustrations of the three scenarios}. The following example illustrates a case where Assumption~\ref{assu:compact} fails in a counterfactual analysis. It also demonstrates scenario (\emph{i}), in which both the identified set $\Theta_I(F)$ and the support-function set $\tilde{\Theta}(F)$ are unbounded.

\begin{example}[continues=ex:entry_game]
Let us revisit the entry game, focusing on the baseline model governed by the moment restrictions in \eqref{eq:entry_medium}. Because these moments are constructed from indicator differences, the set $\mathcal{R}(z;\theta)$ defined in \eqref{eq:moment_random_set} contains a finite number of elements (and hence is closed) for any fixed $(z,\theta)$. Moreover, because $r(u,z;\theta)$ is constructed from a bounded scalar function multiplied by the covariates $x$, there exists a constant $C\in(0,\infty)$ such that for all realizations $z=(x,y)$,
\[
\sup\big\{\norm{t}:\ t\in\mathcal{R}(z;\theta)\big\} \le C\norm{x}.
\]
Therefore, provided $\expt_F[\norm{X}]<\infty$, Assumption~\ref{assu:compact} holds for the baseline structural model.

However, Assumption~\ref{assu:compact} fails once the model is augmented to incorporate counterfactual parameters that depend on the levels of latent shocks, such as expected total firm surplus or expected profit conditional on entry. To be concrete, consider a counterfactual involving a change in profit shifters, $\tilde{X} = \phi(X)$, and suppose the parameter of interest $\tilde{\theta}$ is the expected total firm surplus in the counterfactual environment. In this case, the stacked moment vector $r'$ in the augmented model includes a component defining $\tilde{\theta}$ that is affine in the payoff shocks:
\begin{equation*}
r'_k(u',z;\theta') = \sum_{j} \tilde{y}_j \Big[ \tilde{x}_j^\top \alpha - \Delta_j \tilde{y}_{1-j} + u_j \Big]  - \tilde{\theta}.
\end{equation*}
Under the equilibrium support restriction, the latent shocks $u_j$ are bounded from only one side. For example, if $y_j=1$, the best-response inequality implies $u_j \ge -x_j^\top \alpha + \Delta_j y_{1-j}$, which places no upper bound on $u_j$. Consequently, the set $\mathcal{R}'(z;\theta')$ for the augmented model (defined using $r'$ and $\Gamma'$) is unbounded on events with positive probability. Thus, Assumption~\ref{assu:compact} fails for the augmented model in this counterfactual exercise.

In Appendix~\ref{sec:proof_claim_example}, I show that the support-function set for the augmented model, denoted by $\tilde{\Theta}'(F)$, yields only a one-sided bound for the counterfactual parameter $\tilde{\theta}$. Specifically, if $\prob_F(Y_j=1)>0$ for some firm $j$, the support-function approach does not deliver a finite upper bound on counterfactual total firm profits. At the same time, after analyzing the structure model, I also show that the identified set for counterfactual total firm profits likewise has no finite upper bound. Thus, the absence of a finite upper bound is not a failure of the support-function approach. Rather, it is a consequence of the maintained model assumptions.

This result has an important practical implication. To obtain a finite upper bound on counterfactual total firm surplus, an applied researcher must impose stronger economic assumptions on the structural model rather than seek an alternative partial-identification method. In this example, I establish this conclusion by explicitly analyzing the structure model and its identified set. 
However, as established later in Theorem~\ref{thm:irreducible_model_id}, whenever the structural model satisfies a general irreducibility condition, one can show \textit{a priori}, without resorting to detailed mathematical derivations, that the identified set and the support-function set cannot be distinguished statistically in finite samples.
\end{example}

Next, I illustrate scenario (\emph{ii}), demonstrating that Assumption~\ref{assu:compact} can fail even when the maintained model possesses sufficient identifying power to yield finite two-sided bounds for both the identified set $\Theta_I(F)$ and the support-function set $\tilde{\Theta}(F)$.

\begin{example}[continues=ex:entry_game]
Consider again the entry game, but suppose the researcher instead imposes the moment restrictions specified in \eqref{eq:entry_game_uncorr} and \eqref{eq:entry_game_varaince}. Under this specification, the latent payoff shocks $U_j$ are assumed to be mean-zero, uncorrelated with the covariates $X$, and to have a variance equal to $\tau_j \sigma^2_{R,j} \le \sigma^2_{R,j}$.

Unlike the specification based on median restrictions, Assumption~\ref{assu:compact} fails here even for the baseline structural model. To see why, note that the moment vector $r(u,z;\theta)$ now includes components that are linear and quadratic in $u_j$: namely, $u_j$, $u_j x$, and $u_j^2 - \tau_j \sigma^2_{R,j}$. As discussed previously, the equilibrium support restriction bounds $u_j$ from only one side. For instance, when firm $j$ enters ($y_j=1$), the structural model requires $u_j \ge -x_j^\top \alpha + \Delta_j y_{1-j}$, leaving $u_j$ unbounded from above. Because this one-sided bound allows $u_j \to +\infty$, the quadratic term $u_j^2$ inside the moment vector grows without bound. Consequently, for any given realization $z$, the set $\mathcal{R}(z;\theta)$ is unbounded, causing a failure of Assumption~\ref{assu:compact} for the structural model. As a result, Assumption~\ref{assu:compact} also fails for any augmented model designed for counterfactual analysis.

Despite the failure of Assumption~\ref{assu:compact}, the model retains substantial empirical content. Consider the same counterfactual exercise as before, focusing on the expected counterfactual profit. Because the baseline moment restrictions explicitly constrain the second moment of the latent shocks via $\expt_F[U_j^2] = \tau_j \sigma^2_{R,j}$, we know the expectation $\expt_F[|U_j|]$ has a finite bound. Consequently, one can show that both the identified set $\Theta'_I(F)$ and the support-function set $\tilde{\Theta}'(F)$ for the augmented model imply a finite two-sided interval for the expected counterfactual profit.

This example illustrates that a failure of the integrable boundedness condition required by Assumption~\ref{assu:compact} does \emph{not} imply a lack of identifying power. A structural model can possess an unbounded set $\mathcal{R}(z;\theta)$, causing Assumption \ref{assu:compact} to break down, while still containing enough empirical content to generate nontrivial two-sided bounds for economically important counterfactual parameters.
\end{example}

Assumption~\ref{assu:compact} fails in the two preceding examples for a common reason: the support restriction delivers only a one-sided bound on the latent variables, while the moment functions depend on those latent variables in levels or higher-order terms. As a result, the random set $\mathcal{R}(z;\theta)$ need not be integrably bounded. Example~\ref{ex:production_function} falls into the same category, although I omit the details for brevity.

The next example illustrates a different source of failure of Assumption~\ref{assu:compact}. In that case, the problem arises not from one-sided support restrictions on latent variables, but from the way the model is represented: a support restriction is buried implicitly in the moment restrictions. It also illustrates scenario (\emph{iii}), in which the support-function set $\tilde{\Theta}(F)$ is completely uninformative even though the identified set $\Theta_I(F)$ remains informative. For simplicity, I use Example~\ref{ex:interval_censored_regression} to illustrate these points.

\begin{example}[continues=ex:interval_censored_regression]
Consider first the formulation in Example~\ref{ex:interval_censored_regression}, which imposes the interval condition through the support restriction \eqref{eq:interval_regression_support}. For each realized $z=(\underline{y},\overline{y},w)$, the section $\Gamma(z;\theta)=[\underline{y},\overline{y}]$ is compact. Because $r(u,z;\theta)$ is continuous in $u$, the set $\mathcal{R}(z;\theta)$ is therefore compact. Under mild moment conditions on the observables, such as
\[
\expt_F\big[|W|^2+|\underline{Y}|^2+|\overline{Y}|^2\big]<\infty,
\]
Assumption~\ref{assu:compact} holds.

Now consider a reformulation that treats the interval restriction as an additional moment condition:
\begin{equation}\label{eq:interval_regression_support_weak}
\Gamma^\dagger(\theta)=\mathcal{U}\times\mathcal{Z}
\end{equation}
and
\begin{equation}\label{eq:interval_regression_moment_weak}
r^\dagger(u,z;\theta)
=
\big(u-\alpha-\beta w,\ w(u-\alpha-\beta w),\ \indicator\{\underline{y}\le u\le \overline{y}\}-1\big)^\top.
\end{equation}
Let $\Theta_I^\dagger(F)$ and $\tilde{\Theta}^\dagger(F)$ denote, respectively, the identified set and the support-function set under the reformulated model $(\Gamma^\dagger,r^\dagger)$. Because $u$ is unrestricted under $\Gamma^\dagger(\theta)$, the set
\[
\mathcal{R}^\dagger(z;\theta)
=
\Big\{
\big(u-\alpha-\beta w,\ w(u-\alpha-\beta w),\ \indicator\{\underline{y}\le u\le \overline{y}\}-1\big)^\top:\ u\in\real
\Big\}
\]
is unbounded. Hence, Assumption~\ref{assu:compact} fails under this formulation, even though it imposes the exact same economic restrictions as the original model.

The support function of $\mathcal{R}^\dagger(z;\theta)$ can be computed explicitly. For any direction $\lambda=(\lambda_1,\lambda_2,\lambda_3)^\top\in\mathcal{S}$, factoring the terms linear in $u$ yields:
\begin{align*}
\gamma^\dagger(\lambda,z;\theta)
&=
\sup_{u\in\real}
\Big\{
(\lambda_1+\lambda_2 w)(u-\alpha-\beta w)
+\lambda_3\big(\indicator\{\underline{y}\le u\le \overline{y}\}-1\big)
\Big\}\\
&=
\begin{cases}
+\infty, & \text{if } \lambda_1+\lambda_2 w \neq 0,\\[4pt]
\max\{0,-\lambda_3\}, & \text{if } \lambda_1+\lambda_2 w = 0.
\end{cases}
\end{align*}
In particular, $\gamma^\dagger(\lambda,z;\theta)\ge 0$ for all $(\lambda,z,\theta)$. Therefore, $\tilde{\Theta}^\dagger(F)=\Theta$. In other words, under this reformulation, the support-function approach is completely uninformative.

By contrast, the identified set $\Theta_I^\dagger(F)$ is informative. For example, because $\expt_F[Y^*]=\alpha+\beta\,\expt_F[W]$ and $\underline{Y}\le Y^*\le \overline{Y}$ almost surely, any parameter $\theta = (\alpha, \beta)^\top \in \Theta_I^\dagger(F)$ must satisfy the elementary bounds:
\[
\expt_F[\underline{Y}]
\;\le\;
\alpha+\beta\,\expt_F[W]
\;\le\;
\expt_F[\overline{Y}].
\]
Hence, $\Theta_I^\dagger(F)$ is a strict subset of $\Theta$. 

More broadly, this stylized example highlights two key points. First, the validity of Assumption~\ref{assu:compact} can depend on how the model is represented, in particular on whether support restrictions are imposed pointwise or absorbed into the moment vector. Second, the failure of Assumption~\ref{assu:compact}---in this case, due to mixing the support restriction and the moment restriction---can lead to a significant divergence between the identified set and the support-function set. This naturally raises a broader question: is the divergence between the identified set and the support-function set always due to mixing the support and moment restrictions? Do there exist other cases that could cause a substantial difference between these two sets? The theoretical results in Section~\ref{subsec:finite_sample_indistinguishability} provide a formal treatment of these questions.
\end{example}

I conclude this section with a brief remark on an  algebraic simplification of the support-function inequalities.
\begin{remark}
In some applications, certain components of the moment vector $r(u,z;\theta)$ depend only on the observables $z$ and not on the latent variables $u$. In such situations, the support-function characterization can be simplified. Suppose we partition the moment vector as $r(u, z; \theta) = \big(r_1(z; \theta)^\top, r_2(u, z; \theta)^\top\big)^\top$, where $r_1$ collects the components of $r$ that do not depend on $u$. Then, the moment-inequality restriction in \eqref{eq:moment_ineq} is equivalent to the following paired conditions:
\[
\expt_F\!\big[r_1(Z;\theta)\big]=0
\qquad\text{and}\qquad
\inf_{\lambda\in\mathcal{S}_2}\expt_F\!\big[\gamma_2(\lambda,Z;\theta)\big]\ge 0,
\]
where $\mathcal{S}_2\equiv\{\lambda\in\real^{\dim(r_2)}:\ \norm{\lambda}=1\}$ and $\gamma_2(\lambda,z;\theta) \equiv \sup\{\lambda^\top r_2(u,z;\theta):$ $(u,z)\in \Gamma(\theta)\}$. That is, the continuum of support-function inequalities need only be evaluated over the $u$-dependent subvector $r_2$, while the $u$-independent subvector $r_1$ can be evaluated as standard moment equalities.
\end{remark}

\subsection{Moment closure of the identified set}\label{sec:moment_closure}

This subsection studies what the support-function approach characterizes when Assumption~\ref{assu:compact} is not imposed. As discussed above, the sharp characterization in Theorem~\ref{thm:sharp_support_function} relies on integrable boundedness, a condition that can fail in many structural and counterfactual applications. I therefore present a characterization that requires only the minimal regularity conditions in Assumption~\ref{assu:reg}. To state the result, I first define the \emph{moment closure} of the identified set.

\begin{definition}[Moment closure of the identified set]
\label{def:moment_closure_id_set}
For any $F\in\mathcal{F}$, the moment closure of the identified set, denoted by $\overline{\Theta}_I(F)$, is defined as
\[
\overline{\Theta}_I(F)
\;\equiv\;
\Big\{
\theta\in\Theta:\ 
\inf_{H\in\mathcal{H}(\theta,F)}\big\|\expt_H[r(U,Z;\theta)]\big\| = 0
\Big\},
\]
where $\|\cdot\|$ is any norm on $\real^{d_r}$, and the infimum over an empty set is understood to be $+\infty$. Equivalently, $\theta\in\overline{\Theta}_I(F)$ if, for every $\epsilon>0$, there exists some $H\in\mathcal{H}(\theta,F)$ such that
\[
\big\|\expt_H[r(U,Z;\theta)]\big\|\le \epsilon.
\]
\end{definition}

Comparing Definitions~\ref{def:id_set} and \ref{def:moment_closure_id_set}, the difference is that $\Theta_I(F)$ requires the existence of an admissible joint distribution $H$ that satisfies the moment restriction exactly, whereas $\overline{\Theta}_I(F)$ requires only that the moment violation can be made arbitrarily small. By construction, $\Theta_I(F)\subseteq \overline{\Theta}_I(F)$.

The next theorem shows that, under Assumption~\ref{assu:reg} alone, the support-function inequalities characterize the moment closure $\overline{\Theta}_I(F)$.

\begin{theorem}\label{thm:weak_sharp_support_function}
Under Assumption~\ref{assu:reg}, for any $F\in\mathcal{F}$ and any $\theta\in\Theta$, $\theta\in\overline{\Theta}_I(F)$ if and only if $\theta$ satisfies \eqref{eq:moment_ineq}. Equivalently, for any $F\in\mathcal{F}$,
\[
\overline{\Theta}_I(F)=\tilde{\Theta}(F).
\]
\end{theorem}

Theorem~\ref{thm:weak_sharp_support_function} provides a fundamental characterization of the support-function approach under minimal regularity conditions: the support-function approach is always sharp for the moment closure of the identified set. This result has two important implications. First, whenever the identified set $\Theta_I(F)$ differs from the support-function set $\tilde{\Theta}(F)$, the discrepancy is exactly the gap between the identified set and its moment closure. In particular, this gap does not arise because the support-function approach introduces an artifactual source of slackness beyond the moment closure itself. Theorem~\ref{thm:weak_sharp_support_function} thus shifts the analytical focus from the relationship between $\Theta_I(F)$ and $\tilde{\Theta}(F)$ to the relationship between $\Theta_I(F)$ and $\overline{\Theta}_I(F)$. As demonstrated below, this shift enables the discovery of conditions that are tied to the fundamental geometry structure of the model rather than to the mechanics of the support-function method.

Second, the role of Assumption~\ref{assu:compact} in Theorem~\ref{thm:sharp_support_function} can be understood as providing a sufficient condition under which $\Theta_I(F)$ and $\overline{\Theta}_I(F)$ coincide.  When Assumption~\ref{assu:compact} fails, however, $\Theta_I(F)$ and $\overline{\Theta}_I(F)$ may differ. Indeed, by Theorem~\ref{thm:weak_sharp_support_function}, the example illustrating scenario (\emph{iii}) also shows that the identified set and its moment closure can diverge substantially, even though their definitions differ only in whether the moment restriction must hold exactly or can be approximated arbitrarily well. This observation motivates a closer analysis of the relationship between $\Theta_I(F)$ and $\overline{\Theta}_I(F)$. In particular, it is useful to find conditions under which the gap between the identified set and its moment closure is negligible, while remaining substantially weaker than Assumption~\ref{assu:compact}. This is the focus of the next subsection.

\subsection{Finite-sample indistinguishability between the identified set and its moment closure}
\label{subsec:finite_sample_indistinguishability}
Establishing exact equality between the identified set and its moment closure is often intractable without the restrictive compactness-like conditions of Assumption~\ref{assu:compact}. To bypass this theoretical hurdle, instead of asking when these two sets coincide exactly, I consider a different question in this subsection: under what circumstances is it \emph{impossible} for any inference procedure to distinguish the identified set from its moment closure in finite samples? This is not the standard question in identification analysis, but it is closely aligned with the practical role of identification arguments. Identification is an asymptotic notion, and its usefulness lies in approximating what can be learned from data when the sample size is sufficiently large.  If, for any fixed sample size $n$, no inference procedure can distinguish $\Theta_I(F)$ from $\overline{\Theta}_I(F)$, then the fact that an identification approach characterizes $\overline{\Theta}_I(F)$ rather than $\Theta_I(F)$ is of limited practical consequence. In such cases, the moment closure represents the effective limit of what can be learned from any finite realization of the data.

As will become clear below, finite-sample indistinguishability is ultimately a population-level property. In particular,
whether $\Theta_I(F)$ and $\overline{\Theta}_I(F)$ are distinguishable in finite samples depends on how the model is
formulated, rather than on finite-sample details of any particular estimator or test.

Specifically, I investigate the conditions under which it is impossible to test the null hypothesis that $\theta$ belongs to the identified set against the alternative that it belongs to the moment closure but not the identified set:
\[
H_0: \theta \in \Theta_I(F) \quad \text{vs.} \quad H_1: \theta \in \overline{\Theta}_I(F) \setminus \Theta_I(F).
\]
Let $\{Z_i\}_{i=1}^n$ be an i.i.d. sample drawn from a distribution $F \in \mathcal{F}$. For a fixed $\theta \in \Theta$, define the collections of distributions consistent with $\theta$ under the identified set and its moment closure, respectively:
\[
\mathcal{F}_\theta \coloneqq \{F \in \mathcal{F} : \theta \in \Theta_I(F)\}, \qquad \overline{\mathcal{F}}_\theta \coloneqq \{F \in \mathcal{F} : \theta \in \overline{\Theta}_I(F)\}.
\]
By construction, the null hypothesis $\theta \in \Theta_I(F)$ is equivalent to the statement $F \in \mathcal{F}_\theta$, while the alternative hypothesis $\theta \in \overline{\Theta}_I(F) \setminus \Theta_I(F)$ is equivalent to $F \in \overline{\mathcal{F}}_\theta \setminus \mathcal{F}_\theta$. Since $\Theta_I(F) \subseteq \overline{\Theta}_I(F)$ holds for all $F$, it follows that $\mathcal{F}_\theta \subseteq \overline{\mathcal{F}}_\theta$ for any $\theta \in \Theta$.

Let $\phi_n: \mathcal{Z}^n \to [0, 1]$ denote a (possibly randomized) test function, where $\phi_n(Z_1, \ldots, Z_n)$ represents the probability of rejecting the null hypothesis given the sample. Under $H_0: F \in \mathcal{F}_\theta$, the size of the test $\phi_n$ is defined as the supremum rejection probability over the null:
\[
\sup_{F \in \mathcal{F}_\theta} \expt_F \left[ \phi_n(Z_1, \ldots, Z_n) \right].
\]
The inclusion $\mathcal{F}_\theta \subseteq \overline{\mathcal{F}}_\theta$ implies the following fundamental inequality:
\[
\sup_{F \in \mathcal{F}_\theta} \expt_F \left[ \phi_n(Z_1, \ldots, Z_n) \right] \le \sup_{F \in \overline{\mathcal{F}}_\theta} \mathbb{E}_F \left[ \phi_n(Z_1, \ldots, Z_n) \right].
\]
If this inequality holds with equality, then any test that controls size over the set $\mathcal{F}_\theta$ necessarily yields a rejection probability no greater than that size over the larger set $\overline{\mathcal{F}}_\theta$. Consequently, such a test cannot exhibit nontrivial power against the alternative $H_1: F \in \overline{\mathcal{F}}_\theta \setminus \mathcal{F}_\theta$, which is equivalent to $H_1:\theta \in \overline{\Theta}_I(F) \setminus \Theta_I(F)$. In this scenario, the power of the test is uniformly bounded by its size, rendering the two hypotheses statistically indistinguishable. This observation motivates the following formal definition.

\begin{definition}[Finite-sample indistinguishability]
\label{def:indistinguish}
Fix an arbitrary $n\in\mathbb{N}$. For any $\theta\in\Theta$, the hypothesis $\theta\in \Theta_I(F)$ is said to be
\emph{impossible to distinguish} from the hypothesis $\theta\in \overline{\Theta}_I(F)$ in samples of size $n$ if, for every
test function $\phi_n:\mathcal{Z}^n\to[0,1]$,
\[
\sup_{F\in\mathcal{F}_\theta}\expt_F\big[\phi_n(Z_1,\ldots,Z_n)\big]
=
\sup_{F\in\overline{\mathcal{F}}_\theta}\expt_F\big[\phi_n(Z_1,\ldots,Z_n)\big].
\]
If this holds for every $\theta\in\Theta$ and every $n\in\mathbb{N}$, then the identified set $\Theta_I$ is said to be
\emph{impossible to distinguish} from its moment closure $\overline{\Theta}_I$ in finite samples.
\end{definition}

As it turns out, finite-sample indistinguishability between the identified set and its moment closure can be established under conditions significantly weaker than those required by Assumption~\ref{assu:compact}. To formalize these conditions, I introduce the concepts of a \emph{reduced model} and \emph{irreducibility}. In essence, a model is \emph{reducible} if we can apply a linear rotation to the moment vector such that at least one rotated moment condition implies a deterministic support restriction. 

\begin{definition}[Reduced model, reducibility, and irreducibility]
\label{def:reduced_model}
Fix a model $(\Gamma,r)$ and a parameter value $\theta\in\Theta$, and let $\Theta_I(F;\Gamma,r)$ denote its identified set.
Assume for simplicity that every component of $r(u,z;\theta)$ depends on $u$. A model $(\tilde{\Gamma},\tilde{r})$ is a
\emph{reduced model} of $(\Gamma,r)$ at $\theta$ if the following two conditions hold:
\begin{itemize}
    \item \emph{(Construction):} There exist $d_r$ linearly independent vectors $\lambda_1, \ldots, \lambda_{d_r} \in \mathbb{R}^{d_r}$ such that
    \begin{align*}
    \tilde{\Gamma}(\theta) &\equiv \left\{ (u,z) \in \Gamma(\theta) : \lambda_1^\top r(u,z;\theta) = \gamma(\lambda_1, z; \theta) \right\}, \\
    \tilde{r}(u,z;\theta) &\equiv \left( \gamma(\lambda_1, z; \theta), \lambda_2^\top r(u,z;\theta), \ldots, \lambda_{d_r}^\top r(u,z;\theta) \right)^\top,
    \end{align*}
    where $\gamma(\cdot, \cdot; \theta)$ is the support function defined in \eqref{eq:support_function_def}.
    
    \item \emph{(Identification equivalence):} For every $F \in \mathcal{F}$, $\theta \in \Theta_I(F; \Gamma, r)$ if and only if $\theta \in \Theta_I(F; \tilde{\Gamma}, \tilde{r})$.
\end{itemize}
Model $(\Gamma,r)$ is \emph{reducible} if it admits at least one reduced model at some $\theta\in\Theta$, and it is
\emph{irreducible} otherwise.
\end{definition}

Relative to $(\Gamma,r)$, a reduced model $(\tilde{\Gamma},\tilde{r})$ modifies both the support and the moment
restrictions. On the support side, the restriction is strengthened by requiring the scalar functional
$\lambda_1^\top r(u,z;\theta)$ to attain its support-function value $\gamma(\lambda_1,z;\theta)$ almost surely. On the
moment side, $\tilde{r}$ is obtained by applying a nonsingular linear transformation to the original moment vector
$r(u,z;\theta)$. Because of the modified support restriction $\tilde{\Gamma}(\theta)$, the first rotated moment function
$\lambda_1^\top r(u,z;\theta)$ equals $\gamma(\lambda_1,z;\theta)$ and therefore depends only on $(z,\theta)$
and not on the latent variable $u$. In this sense, a reduced model reallocates part of the original moment restriction on
latent variables into the support restriction, yielding an equivalent but simpler representation in which at least one
moment condition no longer involves latent variables.  For simplicity, Definition~\ref{def:reduced_model} is stated under the restriction that all components of $r$ depend on $u$. A definition that covers the general case is provided in Appendix \ref{sec:distinguish_appendix}.

I illustrate the concept of reducibility using Example~\ref{ex:interval_censored_regression} below.

\begin{example}[continues=ex:interval_censored_regression]
Recall that in Example~\ref{ex:interval_censored_regression}, the interval condition $\underline{Y}\le Y^*\le \overline{Y}$
can be encoded either as a support restriction or as an additional moment restriction. In particular, if one treats the
interval condition as a moment restriction, then one obtains an alternative formulation $(\Gamma^\dagger,r^\dagger)$, defined
in \eqref{eq:interval_regression_support_weak} and \eqref{eq:interval_regression_moment_weak}.

This alternative formulation $(\Gamma^\dagger,r^\dagger)$ is reducible at every $\theta\in\Theta$. To see this, note that $r^\dagger$ has three components, so take the linearly independent vectors
\[
\lambda_1=(0,0,1)^\top,\qquad \lambda_2=(1,0,0)^\top,\qquad \lambda_3=(0,1,0)^\top.
\]
The reduced model constructed from $(\Gamma^\dagger,r^\dagger)$ using $(\lambda_1,\lambda_2,\lambda_3)$ imposes the
support condition $\lambda_1^\top r^\dagger(U,Z;\theta)=\gamma(\lambda_1,Z;\theta)$, which is equivalent to requiring
$U$ (i.e., $Y^*$) to lie in the observed interval $[\underline{Y},\overline{Y}]$. Moreover, the first component of the reduced moment
vector becomes $\gamma(\lambda_1,z;\theta)$, which is identically zero and therefore redundant. The remaining two components
coincide with the original orthogonality conditions in \eqref{eq:interval_regression_moment}. Hence, the reduced model
coincides with the original model $(\Gamma, r)$ of Example~\ref{ex:interval_censored_regression} (up to an additional
redundant moment condition), and the two formulations have the same identified set.
\end{example}

The example illustrates the intuition behind reducibility. If a model is reducible, then some of its moment restrictions
implicitly encode additional support restrictions. Conversely, irreducibility means that no further reallocation of this
type is possible. The next theorem shows that, for irreducible
models, the identified set and its moment closure are indistinguishable in finite samples.

\begin{theorem}
\label{thm:irreducible_model_id}
Suppose Assumption~\ref{assu:reg} holds and $\mathcal{F}$ is convex. If $(\Gamma,r)$ is irreducible, then it is impossible to
distinguish the identified set and its moment closure in finite samples.
\end{theorem}

Irreducibility is a property of the \emph{representation} of a model, rather than of the underlying economic structure. In
particular, a given structural model can always be reformulated by moving any implicit support implications buried in the moment
conditions into explicit support restrictions. In this sense, Theorem~\ref{thm:irreducible_model_id} delivers a highly constructive message: as long as the researcher writes the model in an irreducible form, the identified set and its moment closure cannot be distinguished in finite samples.

Combining Theorems~\ref{thm:weak_sharp_support_function} and \ref{thm:irreducible_model_id} confirms that the support-function approach provides a robust, sensible characterization of identification even when Assumption~\ref{assu:compact} fails. Under the minimal regularity conditions of Assumption~\ref{assu:reg}, the support-function inequalities are sharp for the moment closure. Under irreducibility, the theoretical gap between the moment closure and the identified set is practically irrelevant in finite samples. Together, these results significantly extend the identification theory for the support-function approach, justifying its use in complex structural models and counterfactuals beyond the restrictive bounds of Assumption~\ref{assu:compact}.

These findings also clarify the interpretation of alternative methods. For example, \cite{schennach_entropic_2014} proposes an
entropic approach that targets the moment closure rather than the identified set itself. Theorem~\ref{thm:irreducible_model_id}
implies that, under irreducibility, the difference between these two targets is negligible from a finite-sample perspective.

\section{Discussion}\label{sec:conclusion}
Counterfactual analysis in incomplete models requires a different perspective than the conventional
``estimate--then--simulate'' workflow used in point-identified structural settings. For incomplete models characterized by
support and moment restrictions, this paper argues that identifying structural parameters and conducting counterfactual
analysis are isomorphic tasks. Often, counterfactual exercise can be incorporated into the original structural specification
through an augmented model that treats counterfactual parameters on the same footing as structural parameters. As a result,
both types of parameters can be analyzed within a unified identification framework, without relying on simulation of
counterfactual outcomes from a set-valued prediction rule.

To make this unified approach operational, the paper extends identification theory for the support-function approach beyond
the scope of existing sharp results. Classical sharp characterizations typically rely on integrable boundedness, an
assumption that is often violated in empirically relevant counterfactual exercises, especially when economically important
targets such as profits, surplus, or welfare are unbounded. I show that, under minimal regularity conditions, the
support-function inequalities remain sharp for the moment closure of the identified set. I then establish conditions
under which the identified set and its moment closure are indistinguishable in finite samples in the sense that no
finite-sample inference procedure can separate them. This result provides a practical justification for applying the
support-function approach in settings where integrable boundedness fails, thereby extending its usefulness to a broad class
of structural models and counterfactual questions.

Beyond the specific results in this paper, this finite-sample perspective suggests a broader view for identification analysis. The literature typically treats sharp identification of the identified set as the benchmark, often at the cost of strong regularity conditions, which could be restrictive and are sometimes hard to verify.  An alternative benchmark is to characterize the set of parameters that is indistinguishable to the identified set in finite samples. Shifting attention toward such finite-sample indistinguishability targets may allow for substantially weaker assumptions and more generally applicable identification results, including in environments not considered here. Developing this perspective further is a promising direction for future research.

\newpage

\begin{appendix}

\section{Preliminaries on Random Set and Measurable Selection}
\subsection{Basic results on random set}
This subsection collects some basic concepts and results of random sets that I use to prove results in the paper. Throughout the paper, the random set is defined on a finite-dimensional Euclidean space.  I follow the notation in \cite{molchanov_theory_2005} whenever possible. 

I first introduce some basic concepts formally.
\begin{defsec}[Random Set]\label{def:random_set}
  Let $(\Omega,\mathscr{S}, P)$ be a probability space. A correspondence
  $Y:\Omega\rightrightarrows \real^d$ is said to be a \emph{random closed set}
  if (\emph{i}) $Y(\omega)$ is closed almost surely; (\emph{ii}) for each
  compact set $K$ in $\real^d$, $\{\omega\in \Omega: Y(\omega) \cap K\neq \emptyset\} \in
  \mathscr{S}$.
\end{defsec}

Fix a complete probability space $(\Omega,\mathscr{S}, P)$. Let $L^1(\Omega;
\real^d)$ denote the set of all integrable functions $f:\Omega\mapsto \real^d$.
The following introduces the expectation concept of random set theory.

\begin{defsec}[integrable selections]\label{def:int_sel_awofei}
  If $Y$ is a random closed set, then $S^1(Y)$ denotes the family of all
  integrable selections of $Y$. That is,
  \begin{displaymath}
    S^1(Y) \coloneqq \{f \in L^1(\Omega; \real^d): f(\omega) \in Y(\omega) \text{ almost surely} \}
  \end{displaymath}
\end{defsec}

\begin{defsec}[integration of random set]
  Let $Y$ be a random closed set. Its \emph{Aumann integral} $\expt_I Y$ is
  defined as the set of all expectations of integrable selections,
  \begin{displaymath}
    \expt_I Y \coloneqq \{\expt f: f\in S^1(Y)\}
  \end{displaymath}
  Its \emph{selection expectation} $\expt Y$ is defined as the closure of $\expt_I
  Y$,
  \begin{displaymath}
    \expt Y \coloneqq \cl\{\expt f: f\in S^1(Y)\}
  \end{displaymath}
\end{defsec}

\begin{defsec}[integrable random set]\label{def:integrable_bound}
  A random closed set $Y$ is called \emph{integrable} if $S^1(Y)\neq \emptyset$.
  A random closed set $Y$ is called \emph{integrably bounded} if $\norm{Y}
  \coloneqq \sup\{\norm{t}: t\in Y\} $ has finite expectation, i.e. $\norm{Y}\in
  L^1(\Omega; \real)$.
\end{defsec}

The following lemma summarizes the results on random sets that I used to prove the theorems in the paper.
\begin{lemmasec}\label{lem:random_set_all}
  Let $Y$ be a random closed set, whose realization is a subset of $\real^d$.
  \begin{enumerate}
  \item \label{enu:wa1}$S^1(Y) \neq \emptyset$ if and only if $\inf\{\norm{t}:
    t\in Y\}$ is integrable.
    \item \label{enu:wa2} If $Y$ is integrably bounded, $\expt_I Y$ is a compact
      set and $\expt Y = \expt_I Y$.
    \item \label{enu:wa3} If a function $\zeta:\real^d \mapsto
      \real\cup\{\pm\infty\}$ is upper or lower semicontinuous , then
      $\inf\{\zeta(t): t\in Y\}$ is a random variable. Moreover, if
      $S^1(Y)\neq \emptyset$ and $\expt \zeta(f)$ is defined for all $f\in
      S^1(Y)$ and $\expt \zeta(f) < \infty$ for at least one $f\in S^1(Y)$,
      then
      \begin{displaymath}
        \inf_{f \in S^1(Y)} \expt \zeta(f) = \expt \inf_{t \in
          Y}\zeta(t)
    \end{displaymath}

    \item \label{enu:wa4} If $S^1(Y) \neq \emptyset$,  then $\expt \clco(Y) = \clco \expt Y$ where $\clco$ stands
      for the closure of the convex hull.
  \end{enumerate}
\end{lemmasec}
\begin{proof}
  For results \ref{enu:wa1}, \ref{enu:wa3} and \ref{enu:wa4}, see
  \cite{molchanov_theory_2005}, Theorem 1.7 (p.149), Theorem 1.10 (p. 150) and
  Theorem 1.17 (p. 154) respectively.

  For result \ref{enu:wa2}, Theorem 1.24 on page 158 in
  \cite{molchanov_theory_2005} implies $\expt_I Y$ is a closed set. Moreover,
  since $\norm{v} \le \expt \norm{Y}, \ \forall v\in \expt_I Y$, 
  $\expt_I Y$ is bounded. Since $\expt_I Y \subseteq \real^d$, $\expt_I Y$ is compact.
\end{proof}

\subsection{Selection Theorem}\label{sec:selection}
This subsection collects some concepts and results on measurable selection which will be cited later in the proof.

\begin{defsec}[universally measurable set]
  Let $S$ be a Polish space and let $\mathscr{B}_S$ be its Borel sigma algebra. A
  subset $S'$ of $S$ is a \emph{universally measurable} set if for any complete
  probability space $(S, \mathscr{F}, F)$ with
  $\mathscr{B}_S\subseteq\mathscr{F}$, $S'\in \mathscr{F}$.
\end{defsec}

\begin{defsec}[universally measurable function]\label{def:universally_measurable}
  Let $S$ be a Polish space and let $\mathscr{B}_S$ be its Borel sigma algebra,
  and $T$ be some topological space. A function $f:S\mapsto T$ is
  \emph{universally measurable} if for any Borel set $B$ of $T$, $\{s\in S:
  f(s)\in B\}$ is universally measurable.
\end{defsec}

By definition, if a function is universally measurable, then it's also measurable in the completion of any Borel
probability space. Moreover, any Borel set in a Polish space is universally measurable.
Given $D\subseteq S \times T$, define $\text{proj}_S(D)\coloneqq\{s\in S:
\exists t\in T, (s,t)\in D\}$ and $D_s\coloneqq\{t\in T: (t,s)\in D\}$. The
following lemma is a simplified version of Proposition 7.50(b) in
\cite{bertsekas_stochastic_1978}.

\begin{lemmasec}[measurable selection]\label{lem:selection}
  Let $S$ and $T$ be Polish spaces, let $D\subseteq S \times T$ be a Borel set,
  and let $f:D\to \real$ be a Borel measurable function. Define
  $f^*:\text{proj}_S(D)\to \real\cup\{-\infty\}$ by
  \begin{displaymath}
    f^*(s) = \inf_{t\in D_s} f(s,t).
  \end{displaymath}
  Suppose $f^*(s) > -\infty$ for any $s\in \text{proj}_S(D)$. Then,
the set
    \begin{displaymath}
      I \coloneqq \{s \in \text{proj}_S(D): \exists t_s\in D_s, f(s,t_s) = f^*(s)\}
    \end{displaymath}
    is universally measurable. And, for every $\epsilon > 0$, there exists a
    universally measurable function $\phi:\text{proj}_S(D)\mapsto T$ such that
    (i) $Gr(\phi)\subseteq D$; (ii) for all $s\in \text{proj}_S(D)$, $f(s,\phi(s)) \le f^*(s) + \epsilon, \ \forall s\in S$
    and, (iii) for all $s\in I$, $f(s,\phi(s)) = f^*(s)$.
\end{lemmasec}
\begin{proof}
  Since
  \begin{itemize}
  \item every Borel set is an analytic set,
  \item every Polish space is a
  Borel space as defined in Definition 7.7 in
  \cite{bertsekas_stochastic_1978} (page 118),
  \item every Borel measurable function is lower semianalytic function as
    defined in Definition 7.21 in 
  \cite{bertsekas_stochastic_1978} (page 177),
  \end{itemize}
  the result follows from Proposition 7.50(b) on page 184 in \cite{bertsekas_stochastic_1978}. 
\end{proof}

\section{Proof for Theorems \ref{thm:sharp_support_function} and \ref{thm:weak_sharp_support_function}}
To state the proof, I also need the following extra notation: 

\paragraph{Notation}
For any $F\in \mathcal{F}$, let $\widetilde{\Theta}(F)$ be the set of all $\theta$ which satisfies
\eqref{eq:moment_ineq}.
For any set $A$ in an Euclidean space, I use $\interior A$ to denote its interior
, $\cl A$ to denote its closure, $\co A$ to denote its convex hull and $\clco A$ to denote the closure of its convex
hull.  Given any topological space $X$, let $\mathscr{B}_{\mathcal{X}}$ denote all Borel sets on $\mathcal{X}$, and $\mathscr{P}_{\mathcal{X}}$ denote the set of all probability measures on measurable space $(\mathcal{X}, \mathscr{B}_{\mathcal{X}})$. Recall that $\mathcal{U}$ and $\mathcal{Z}$ denote the space of $U$ and $Z$ respectively. 
For any $F\in \mathcal{F}$, let the probability space $(\mathcal{Z}, \mathscr{Z}, F)$ be the completion of
$(\mathcal{Z}, \mathscr{B}_{\mathcal{Z}}, F)$. Moreover, recall $\Gamma(z;\theta)\equiv \{u\in \mathcal{U}: (u,z)\in \Gamma(\theta)\}$, and $\mathcal{R}(z;\theta)$ is the image of $\Gamma(z;\theta)$ by $r$,
i.e. $\mathcal{R}(z;\theta) \equiv \{r(u,z;\theta): u\in \Gamma(z;\theta)\}$.
Note that \eqref{eq:moment_ineq} can be rewritten as
\begin{displaymath}
\forall \lambda\in \mathcal{S}, \ \expt_F\left[\sup_{t \in \mathcal{R}(Z;\theta)} \lambda't \right] \ge 0.
\end{displaymath}

In the following, I first prove Lemma \ref{lem:suff_random_set} which establishes some useful properties for
$\mathcal{R}(z;\theta)$ as a random set. Then, I would prove Theorem \ref{thm:weak_sharp_support_function}. After that, I would also prove Theorem \ref{thm:sharp_support_function} for completeness.

\subsection{Property of $\mathcal{R}(z;\theta)$}
In the following, I use Assumption \ref{assu:reg}\ref{enu:CC1} and \ref{assu:reg}\ref{enu:CC2} to denote the first and
the second condition in Assumption \ref{assu:reg} respectively. Similarly, I use Assumption 
\ref{assu:compact} \ref{enu:CC3} and Assumption \ref{assu:compact} \ref{enu:CC4} to denote the first and the second
condition in Assumption \ref{assu:compact}. The following lemma provides some basic results needed for the proof of all
theorems.

\begin{lemmasec}\label{lem:suff_random_set}
  Let $F$ be an arbitrary element in $\mathcal{F}$. 
  \begin{enumerate} 
    \item Suppose Assumption \ref{assu:reg}\ref{enu:CC1}
      holds. Then, for each $\theta\in \Theta$, $\cl \mathcal{R}(\cdot;\theta)$ is a random closed set in probability space
      $(\mathcal{Z},\mathscr{Z}, F)$.
\item Suppose Assumption \ref{assu:reg} hold. Then, for each $\theta\in \Theta$, $\cl \mathcal{R}(\cdot;\theta)$ 
  is an integrable random closed set in probability space $(\mathcal{Z},\mathscr{Z}, F)$.
\item Suppose Assumption \ref{assu:reg}\ref{enu:CC1} and Assumption \ref{assu:compact}\ref{enu:CC4}
   hold. Then, for each $\theta\in \Theta$, random closed set $\cl\mathcal{R}(\cdot;\theta)$ is integrably
    bounded in probability space $(\mathcal{Z},\mathscr{Z}, F)$.
  \end{enumerate}
\end{lemmasec}

  \begin{proof}[Proof of Lemma \ref{lem:suff_random_set}]

    (\emph{i}) I first show $\cl\mathcal{R}(\cdot;\theta)$ is a random closed set under Assumption \ref{assu:reg}\ref{enu:CC1}.

  Let $D = \{t_1,t_2,...\}$ be a countable set dense in $\real^{\dim(r)}$. For each
  $t_i\in D$, consider the following optimization problem ,
  \begin{displaymath}
      \inf_{u\in \Gamma(z;\theta)} \norm{t_i - r(u,z;\theta)}
  \end{displaymath}
  Given that $\norm{t_i - r(u,z;\theta)}$ is a Borel measurable function of
  $(u,z)$, that $\Gamma(\theta)$ is a
  Borel set, and that $\Gamma(z;\theta)$ is nonempty almost surely, Lemma
  \ref{lem:selection} implies that, for any $n\in \nature$, there exists a
  universally measurable function $f_{i,n}:\mathcal{Z}\mapsto \mathcal{U}$ such
  that for any $z\in Z$, $f_{i,n}(z)\in \Gamma(z;\theta)$ and
  \begin{displaymath}
    \norm{t_i-r(f_{i,n}(z),z;\theta)} \le \frac{1}{n} +  
    \inf_{u\in \Gamma(z;\theta)} \norm{t_i - r(u,z;\theta)}.
  \end{displaymath}
  See Definition \ref{def:universally_measurable} for the definition of
  a universal measurable function. Since $(\mathcal{Z},\mathscr{Z}, F)$ 
  is the completion of the Borel probability space 
  $(\mathcal{Z},\mathscr{B}_Z, F)$, by the definition of
  universally measurable functions, $f_{i,n}(z)$ is also
  $\mathscr{Z}$-measurable.
    
  Fix an arbitrary $z$. Since, by construction, $f_{i,n}(z)\in \Gamma(z;\theta)$, we
  know $\cl\{r\big(f_{i,n}(z), z\big): i,n\in \nature\} \subseteq
  \cl\mathcal{R}(z;\theta)$. On the other hand, for any $t\in
  \cl\mathcal{R}(z;\theta)$ and any $\epsilon>0$, there must exists some $t_i\in D$
  such that $\norm{t - t_i}\le \epsilon/3$, and there must exists some $n\in
  \nature$ such that $\norm{t_i - r(f_{i,n}(z),z;\theta)} \le 2\epsilon/ 3$.
  Hence, for any $t\in \cl\mathcal{R}(z;\theta)$ and any $\epsilon>0$, there exists
  some $\tilde{t}\in \{r\big(f_{i,n}(z), z\big): i,n\in \nature\}$ such that
  $\norm{t-\tilde{t}}\le \epsilon$. Hence, $\cl\mathcal{R}(z;\theta) =
  \cl\{r\big(f_{i,n}(z),z\big): i,n\in \nature\}$. By Theorem 2.3 on page 26 of
  \cite{molchanov_theory_2005}, $\cl\mathcal{R}(z;\theta)$ is a random closed set
  in $(\mathcal{Z}, \mathscr{Z}, F)$.

  (\emph{ii}) Suppose, in addition, Assumption \ref{assu:reg}\ref{enu:CC2} holds. The fact that
  $\cl\mathcal{R}(z;\theta)$ is a random closed set implies $z\mapsto
  \inf\{\norm{t}: t\in \cl\mathcal{R}(z;\theta)\}$ is measurable in $(\mathcal{Z},\mathscr{Z})$
  (See result \ref{enu:wa3} in Lemma \ref{lem:random_set_all}). Moreover, note that
  \begin{displaymath}
   \inf\{\norm{t}: t\in \mathcal{R}(z;\theta)\} = \inf\{\norm{t}: t\in \cl\mathcal{R}(z;\theta)\}.
 \end{displaymath}
 Assumption \ref{assu:reg}\ref{enu:CC2} then implies $z\mapsto \inf\{\norm{t}: t\in
 \cl\mathcal{R}(z;\theta)\}$ is an integrable function. By Definition
 \ref{def:integrable_bound} and Lemma
 \ref{lem:random_set_all}\ref{enu:wa1}, $\cl\mathcal{R}(\cdot;\theta)$ is
 integrable.

 (\emph{iii}) Finally, given the first result in this lemma, Assumption \ref{assu:compact}\ref{enu:CC4} directly implies
 $\cl\mathcal{R}(\cdot;\theta)$ is integrably bounded by definition. 
\end{proof}

\subsection{Proof for Theorem \ref{thm:weak_sharp_support_function}}\label{sec:proof_general_id}
The proof for Theorem \ref{thm:weak_sharp_support_function} builds on the following two lemmas which I would prove at the end of this subsection.

\begin{lemmasec}\label{lem:convex_lemma}
  Suppose set $A$ is a nonempty closed convex set in $\real^d$. Then $0\in A$ 
  if and only if
  \begin{equation}\label{eq:sep_hyperplane_theorem}
    \inf_{\lambda\in \real^d} \sup\{\lambda't: t\in A\} \ge 0.
  \end{equation}
Note that \eqref{eq:sep_hyperplane_theorem} includes the case that $ \inf_{\lambda\in \real^d} \sup\{\lambda't: t\in A\} =+\infty$ which could happen when $A = \real^d$ .
\end{lemmasec}

\begin{lemmasec}\label{lem:suff_quasi_id}
  Suppose Assumption \ref{assu:reg} hold. Then, for any $F\in \mathcal{F}$,  $0 \in \clco \expt_F 
  \cl\mathcal{R}(Z;\theta)$ implies $\theta \in \overline{\Theta}_I(F)$.
\end{lemmasec}

\begin{proof}[Proof of Theorem \ref{thm:weak_sharp_support_function}]
  Fix an arbitrary element $F$ in $\mathcal{F}$. In the following proof, I will abbreviate $\expt_F$ as $\expt$,
  $\overline{\Theta}_I(F)$ as $\overline{\Theta}_I$, and $\widetilde{\Theta}(F)$ as $\widetilde{\Theta}$.  Recall that
  $\expt_{I}$ stands for the Aumann integral. 

  First of all, I'm going to show $\widetilde{\Theta}(F)\subseteq \overline{\Theta}_I(F)$.  Lemma \ref{lem:suff_random_set}
  implies that $\cl\mathcal{R}(\cdot;\theta)$ is an integrable random closed set in $(\mathcal{Z}, \mathscr{Z}, F)$.
  Suppose, for the purpose of contradiction, there exists $\theta\in \widetilde{\Theta}$ such that $\theta\notin
  \overline{\Theta}_I$. Then, by Lemma \ref{lem:suff_quasi_id}, $0\notin \clco \expt \cl\mathcal{R}(Z;\theta)$. Lemma
  \ref{lem:convex_lemma} then implies  that the following inequality holds:
  \begin{displaymath}
    \inf_{\lambda\in \real^{\dim(r)}} \sup\{\lambda't: t\in \clco \expt
    \cl\mathcal{R}(Z;\theta)\} < 0
  \end{displaymath}

  By Lemma \ref{lem:random_set_all}\ref{enu:wa4}, and the fact that
  $\clco\mathcal{R}(Z;\theta)\subseteq\clco\cl\mathcal{R}(Z;\theta)$, and that
  the Aumann integral 
  $\expt_I\clco\mathcal{R}(Z;\theta)\subseteq \expt\clco\mathcal{R}(Z;\theta)$, we
  know
  \begin{equation}
    \label{eq:temp_1036}
    \inf_{\lambda\in \real^{\dim(r)}} \sup\{\lambda't: t\in \expt_I
    \clco\mathcal{R}(Z;\theta)\} < 0
  \end{equation}
  Choose any $\tilde{\lambda}$ such that $\sup\{\tilde{\lambda}'t: t\in \expt_I
  \clco\mathcal{R}(Z;\theta)\} < 0$. Note that
  \begin{equation}
    \label{eq:amlfewki}
    \sup\{\tilde{\lambda}'t: t\in \expt_I
  \clco\mathcal{R}(Z;\theta)\} = -\inf_{f\in S^1(\clco\mathcal{R}(Z;\theta))} \expt [-\tilde{\lambda}'f]
  \end{equation}
  where $S^1$ is defined in Definition \ref{def:int_sel_awofei}. Apply Lemma
  \ref{lem:random_set_all}\ref{enu:wa3} with $\zeta(t) = -\lambda't$ to get
  \begin{eqnarray}
    && -\inf_{f\in S^1(\clco\mathcal{R}(Z;\theta))} \expt [-\tilde{\lambda}'f] \nonumber\\
    & = & -\expt \inf\{-\tilde{\lambda}'t: t \in \clco\mathcal{R}(Z;\theta)\} \nonumber\\
    & = & \expt \sup\{\tilde{\lambda}'t: t \in \clco\mathcal{R}(Z;\theta)\}. \label{eq:mlkq0w9e}
  \end{eqnarray}
  Equation \eqref{eq:amlfewki} and \eqref{eq:mlkq0w9e} imply
  \begin{equation}
    \label{eq:aoifnewwaefij}
    \expt \sup\{\tilde{\lambda}'t: t \in \clco\mathcal{R}(Z;\theta)\} =
    \sup\{\tilde{\lambda}'t: t\in \expt_I
    \clco\mathcal{R}(Z;\theta)\}  < 0.
  \end{equation}
  In addition, since $\mathcal{R}(z;\theta)\subseteq \real^{\dim(r)}$, 
  \begin{equation}
    \label{eq:temp_1049}
    \sup\{\tilde{\lambda}'t: t\in \clco \mathcal{R}(z;\theta)\} = \sup\{\tilde{\lambda}'t: t\in \mathcal{R}(z;\theta)\},
  \end{equation}
 equation \eqref{eq:aoifnewwaefij} and
  \eqref{eq:temp_1049} imply
  \begin{displaymath}
    \inf_{\lambda\in \real^{\dim(r)}} \expt \sup\{\lambda't: t\in \mathcal{R}(Z;\theta)
    \} < 0.
  \end{displaymath}
  This contradicts $\theta\in \widetilde{\Theta}$. This proves
  $\widetilde{\Theta}\subseteq \overline{\Theta}_I$.

  To show $\overline{\Theta}_I\subseteq \widetilde{\Theta}$. Fix any $\theta\in \overline{\Theta}_I$
  and any $\epsilon>0$, there exists a distribution $H$ of $(U,Z)$ such that (i)
  $\norm{\expt r(U,Z;\theta)} \le \epsilon$; (ii) $\prob_H(U\in \Gamma(Z;\theta)) =
  1$; (iii) the marginal distribution of $H$ on $Z$ equals to $F$. For any 
  $\lambda\in \mathcal{S}$,
  \begin{eqnarray*}
    -\epsilon & \le& \expt_H(\lambda' r(U,Z;\theta)) \\
      & \le & \expt_H\left\{\sup_{u\in \Gamma(Z;\theta)}\lambda'r(u,Z;\theta)\right\} \\
      & =  & \expt\left\{\sup_{u\in \Gamma(Z;\theta)}\lambda'r(u,Z;\theta)\right\} 
  \end{eqnarray*}
  where the first inequality comes from Cauchy-Schwarz inequality, the second
  inequality comes form $\prob_H(U\in \Gamma(Z;\theta)) = 1$, and the last
  equality follows from the fact that $\sup\{\lambda'r(u,z;\theta): 
  u\in \Gamma(z;\theta) \}$ only depends on $z$. Hence,
  \begin{displaymath}
    -\epsilon \le \inf_{\lambda\in \mathcal{S}}\expt\left[\sup_{u\in \Gamma(Z;\theta)}\lambda'r(u,Z;\theta) \right].
  \end{displaymath}
  Since these holds with any $\epsilon > 0$, 
    \begin{displaymath}
      0 \le \inf_{\lambda\in\mathcal{S}}\expt\left[\sup_{u\in \Gamma(Z;\theta)}\lambda'r(u,Z;\theta) \right],
  \end{displaymath}
  which implies  $\theta\in \widetilde{\Theta}$.
  \end{proof}

\begin{proof}[Proof for Lemma \ref{lem:convex_lemma}]
  This is a classic result of the support function. See, for example, Theorem 2.2.2 in
  \cite{hiriart-urruty_fundamentals_2001} for its proof.
\end{proof}

\begin{proof}[Proof for Lemma \ref{lem:suff_quasi_id}]
  Fix an arbitrary $F\in \mathcal{F}$. In the following proof, I will abbreviate $\expt_F$ as $\expt$, and
  $\overline{\Theta}_I(F)$ as $\overline{\Theta}_I$.  Under Assumption \ref{assu:reg}, $\cl\mathcal{R}(Z;\theta)$ is an integrable random
  closed set in $(\mathcal{Z}, \mathscr{Z}, F)$.  Suppose $0 \in \clco \expt \cl\mathcal{R}(Z;\theta)$ is true, I want to
  prove that $\theta\in \overline{\Theta}_I$. 
  
  Fix an arbitrary $\epsilon>0$. By the fact that $\clco A = \clco\cl A$ 
  for any subset
  $A$ in finite dimensional Euclidean space, and that $\expt \cl
  \mathcal{R}(Z;\theta) = \cl(\expt_I \cl\mathcal{R}(Z;\theta))$ by Definition
  \ref{def:integrable_bound}, $0\in \clco \expt \cl \mathcal{R}(Z;\theta)$
  must imply $0 \in \clco \expt_I \cl \mathcal{R}(Z;\theta)$. Hence, there exists 
  some $v\in \co \expt_I \cl \mathcal{R}(Z;\theta)$ such that $\norm{v} \le 
  \epsilon$.  
  By Carathéodory's
  theorem, there must exists $p_0,p_1,...,p_{\dim(r)} \in [0,1]$ and
  $v_0,...,v_{\dim(r)}\in \expt_I \cl \mathcal{R}(Z;\theta)$ such that
  $\sum_{j=0}^{\dim(r)}p_j = 1$ and $v = \sum_{j=0}^{\dim(r)}p_jv_j$.For each 
  $j=0,...,\dim(r)$, there exists $f_j\in S^1(\cl\mathcal{R}(Z;\theta))$ such that $v_j
  = \expt f_j(Z)$. Hence, 
  \begin{displaymath}
    \norm{\sum_{j=0}^{\dim(r)} p_j \expt f_j(Z) } \le \epsilon.
  \end{displaymath}
  By the definition of $S^1(\cl\mathcal{R}(Z;\theta))$,
  each $f_j$ is measurable and integrable in $(\mathcal{Z},\mathscr{Z}, F)$.

  Let $T$ be a random variable independent of $Z$, which is supported on $\{0,
  1, ..., \dim(r)\}$ and is distributed as the following,
  \begin{displaymath}
    \prob(T = j) = p_j,\ \forall j \in \{0, 1, ..., \dim(r)\}.
  \end{displaymath}
  Construct random variable $R\in \real^{\dim(r)}$ from $T$ and $Z$ as
  \begin{displaymath}
    R = \sum_{j=0}^{\dim(r)}\indicator\{T = j\}f_j(Z).
  \end{displaymath}

  Let $H'$ denote the joint distribution of $(Z,R)$ in measurable space
  $(\mathcal{Z}\times\real^{\dim(r)}, \mathscr{B}_{\mathcal{Z}\times\real^{\dim(r)}})$. 
  By construction, $H'$'s marginal distribution for $Z$ equals $F$, and
  \begin{displaymath}
    \prob_{H'}(R\in \cl\mathcal{R}(Z;\theta)) = 1.
  \end{displaymath}
  Also,
  \begin{displaymath}
   \norm{\expt_{H'} R} = \norm{\int\expt_{H'}[R|Z=z] \ud F_Z}
     = \norm{\expt
    \sum_{j=0}^{\dim(r)}p_j f_j(Z)} = \norm{\sum_{j=0}^{\dim(r)} p_j \expt f_j(Z) }\le \epsilon.
  \end{displaymath}

  Now consider $H'$ as in the completion of probability space 
  $(\mathcal{Z}\times\real^{\dim(r)},\mathscr{B}_{\mathcal{Z}\times\real^{\dim(r)}}, H')$.
  Since $\prob_{H'}(R\in \cl\mathcal{R}(Z;\theta)) = 1$, the definition of $\mathcal{R}(Z;\theta)$ implies
  \begin{displaymath}
     \prob_{H'}\Big(\inf_{u\in \Gamma(Z;\theta)} \norm{r(u,Z;\theta) - R} = 0\Big) = 1
  \end{displaymath}  
  Since $\{(z,u):u\in \Gamma(z;\theta)\}\times \real^{\dim(r)}$ is a Borel set, and
  that $(u,z,t)\mapsto \norm{r(u,z;\theta) - t}$ is a Borel measurable function
  in $\mathcal{U}\times\mathcal{Z}\times\real^{\dim(r)}$, Lemma \ref{lem:selection}
  in Appendix \ref{sec:selection} implies that there exists a universally
  measurable function $g:\mathcal{Z}\times \real^{\dim(r)}\mapsto\mathcal{U}$, such
  that for any $t\in \real^{\dim(r)}$ and any $z\in \mathcal{Z}$, $g(z,t) \in
  \Gamma(z;\theta)$ and
  \begin{displaymath}
      \norm{r\big(g(z,t), z\big) - t} \le \epsilon + \inf_{u\in \Gamma(z;\theta)} \norm{r(u,z;\theta) - t}.
  \end{displaymath}

  Construct random variable $U = g(Z, R)$. Let $H$ be the joint distribution of 
  $(U, Z)$ in the measurable space $(\mathcal{U\times Z},
   \mathscr{B}_{\mathcal{U\times Z}})$. Then, $\prob_H(U\in \Gamma(Z;\theta))
  = 1$ and 
  \begin{displaymath}
  \prob_H(\norm{r(U,Z;\theta) - R} \le \epsilon) = 1,
  \end{displaymath}
 so that 
  \begin{displaymath}
    \norm{\expt_H r(U,Z;\theta)} \le \epsilon + \norm{\expt_H R} \le 2\epsilon
  \end{displaymath}
  This completes the proof that $\theta\in \overline{\Theta}_I$.
\end{proof}

\subsection{Proof of Theorem \ref{thm:sharp_support_function}}\label{sec:proof_sharp_id}
The proof for Theorem \ref{thm:sharp_support_function} builds on the following lemma whose proof will be presented at the end of this subsection.

\begin{lemmasec}\label{lem:suff_sharp_id}
  Suppose Assumption \ref{assu:reg} and \ref{assu:compact} hold. 
  Then, for any $F\in \mathcal{F}$, $0 \in \clco \expt_F \cl\mathcal{R}(Z;\theta)$ implies $\theta \in \Theta_I(F)$.
\end{lemmasec}

\begin{proof}[Proof for Theorem \ref{thm:sharp_support_function}]
  Fix an arbitrary $F$ in $\mathcal{F}$. 
  Because I have shown in Theorem \ref{thm:weak_sharp_support_function} that $\overline{\Theta}_I(F)
  = \widetilde{\Theta}(F)$, and because $\Theta_I(F)\subseteq \overline{\Theta}_I(F)$,
  I only need to prove  $\widetilde{\Theta}(F) \subseteq \Theta_I(F)$. To show 
  $\widetilde{\Theta}(F) \subseteq \Theta_I(F)$, suppose, for the purpose of
  contradiction, there exists some $\theta\in \widetilde{\Theta}(F)$ such that 
  $\theta\notin \Theta_I(F)$. Then, by Lemma \ref{lem:suff_sharp_id}, $0\notin \clco \expt_F
  \cl\mathcal{R}(Z;\theta)$. Yet, as shown in the proof of Theorem
  \ref{thm:weak_sharp_support_function}, this contradicts the fact that $\theta\in
  \widetilde{\Theta}(F)$.
\end{proof}

\begin{proof}[Proof for Lemma \ref{lem:suff_sharp_id}]
  Fix an arbitrary $F\in \mathcal{F}$. In the following proof, I will abbreviate $\expt_F$ as $\expt$, and $\Theta_I(F)$ as
  $\Theta_I$. Recall also that $\expt_{I}$ stands for the Aumann integral.

  The proof of this lemma is similar to that of Lemma \ref{lem:suff_quasi_id}.
  One only needs to notice that under Assumption \ref{assu:reg} and \ref{assu:compact},
  $0 \in \clco \expt \cl\mathcal{R}(Z;\theta)$ not only implies 
  $0 \in \clco \expt_I \cl \mathcal{R}(Z;\theta)$ but also
  implies $0\in \co \expt_I \mathcal{R}(Z;\theta)$. For clarity, I
  provide the entire proof.
  
  Suppose $0 \in \clco \expt \cl\mathcal{R}(Z;\theta)$, I want to show
  $\theta \in \Theta_I$. First of all, note that 
  $0 \in \clco \expt \cl\mathcal{R}(Z;\theta)$ is equivalent to 
  $0 \in \clco \expt \mathcal{R}(Z;\theta)$ under Assumption \ref{assu:compact}\ref{enu:CC3}. 
  Moreover, Assumption \ref{assu:compact}\ref{enu:CC4} together with Lemma
  \ref{lem:suff_random_set} also implies $\mathcal{R}(Z;\theta)$ is an integrably
  bounded random closed set. By Lemma \ref{lem:random_set_all}\ref{enu:wa2}, 
  $\expt \mathcal{R}(Z;\theta)$ is a compact set and $\expt \mathcal{R}(Z;\theta)
  = \expt_I\mathcal{R}(Z;\theta)$. Since
  $\expt\mathcal{R}(Z;\theta)\subseteq\real^{\dim(r)}$, Carathéodory's theorem implies
  $\co \expt\mathcal{R}(Z;\theta)$ is also compact. Hence, 
  $0 \in \clco \expt \cl\mathcal{R}(Z;\theta)$ implies $0\in \co
  \expt_I \mathcal{R}(Z;\theta)$.

  Given $0\in \co\expt_I \mathcal{R}(Z;\theta)$, 
  Carathéodory's theorem also implies that there must exists $p_0,p_1,...,p_{\dim(r)}
  \in [0,1]$ and $v_0,...,v_{\dim(r)}\in \expt_I  \mathcal{R}(Z;\theta)$ such that
  $\sum_{j=0}^{\dim(r)}p_j = 1$ and $ \sum_{j=0}^{\dim(r)}p_jv_j = 0 $.

  For each $j=0,...,\dim(r)$, there exists $f_j\in S^1(\mathcal{R}(Z;\theta))$ such
  that $v_j = \expt f_j(Z)$. Hence,
  \begin{displaymath}
    \sum_{j=0}^{\dim(r)} p_j \expt f_j(Z)  = 0.
  \end{displaymath}
  Recall that $(\mathcal{Z},\mathscr{Z}, F)$ denotes the completion of Borel probability space
  $(\mathcal{Z},\mathscr{B}_{\mathcal{Z}}, F)$. By the definition of 
  $S^1(\mathcal{R}(Z;\theta))$, each $f_j$ is measurable and integrable in 
  $(\mathcal{Z},\mathscr{Z}, F)$.

  The remainder of the proof is similar to that in Lemma \ref{lem:suff_quasi_id}.
  Let $T$ be a random variable independent of $Z$, which is supported on $\{0,
  1, ..., \dim(r)\}$ and is distributed as the following,
  \begin{displaymath}
    \prob(T = j) = p_j,\ \forall j \in \{0, 1, ..., \dim(r)\}.
  \end{displaymath}
  Construct random variable $R\in
  \real^{\dim(r)} $ from $T$ and $Z$ as
  \begin{displaymath}
    R = \sum_{j=0}^{\dim(r)}\indicator\{T = j\}f_j(Z)
  \end{displaymath}

  Let $H'$ denote the joint distribution of $(Z,R)$ in measurable space
  $(\mathcal{Z}\times\real^{\dim(r)}, \mathscr{B}_{\mathcal{Z}\times\real^{\dim(r)}})$. 
  By construction, $H'$'s marginal distribution for $Z$ equals $F_Z$, and
  \begin{displaymath}
    \prob_{H'}(R\in \mathcal{R}(Z;\theta)) = 1,
  \end{displaymath}
  and
  \begin{displaymath}
    \expt_{H'} R = \int \expt_{H'}[R|Z=z]\ud F_Z(z) = \expt
    \sum_{j=0}^{\dim(r)}p_jf_j(Z) = \sum_{j=0}^{\dim(r)} p_j \expt f_j(Z)  =  0.
  \end{displaymath}

  Now consider $H'$ as in the completion of probability space 
  $(\mathcal{Z}\times\real^{\dim(r)},\mathscr{B}_{\mathcal{Z}\times\real^{\dim(r)}}, H')$.
  Since $\prob_{H'}(R\in \mathcal{R}(Z;\theta)) = 1$, the definition of $\mathcal{R}(Z;\theta)$ implies
  \begin{displaymath}
    \prob_H \Big(\min_{u\in \Gamma(Z;\theta)} \norm{r(u,Z;\theta) - R} = 0 \Big ) = 1.
  \end{displaymath}
  Since $\{(z,u):u\in \Gamma(z;\theta)\}\times \real^{\dim(r)}$ is a Borel set, and
  $(u,z,t)\mapsto \norm{r(u,z;\theta) - t}$ is a Borel measurable function
  in $\mathcal{U}\times\mathcal{Z}\times\real^{\dim(r)}$, Lemma \ref{lem:selection}
  in Appendix \ref{sec:selection} implies that there exists a universally
  measurable function $g:\mathcal{Z}\times \real^{\dim(r)}\mapsto\mathcal{U}$, such
  that, for any $z\in \mathcal{Z}$ and $t\in \real^{\dim(r)}$, $g(z,t) \in
  \Gamma(z;\theta)$. In addition, for any $z\in \mathcal{Z}$ and $t\in
  \real^{\dim(r)}$ which satisfies
  \begin{displaymath}
\inf_{u\in \Gamma(Z;\theta)} \norm{r(u,z;\theta) - t} = \min_{u\in
  \Gamma(z;\theta)}\norm{r(u,z;\theta) - t},
\end{displaymath}
it must be true that
  \begin{displaymath}
    \norm{r(g(z,t), z) - t}  = \min_{u\in
  \Gamma(z;\theta)}\norm{r(u,z;\theta) - t}.
  \end{displaymath}

  Construct random variable $U = g(Z, R)$. Let $H$ be the joint distribution of 
  $(U, Z)$ in the measurable space $(\mathcal{U\times Z},
   \mathscr{B}_{\mathcal{U\times Z}})$. Then, $\prob_H(U\in \Gamma(Z;\theta))
  = 1$ and 
  \begin{displaymath}
  \prob_H(r(U,Z;\theta) = R ) = 1,
  \end{displaymath}
 so that 
  \begin{displaymath}
    \expt_H r(U,Z;\theta) = \expt_H R = 0
  \end{displaymath}
  This completes the proof that $\theta\in \Theta_I$.
\end{proof}

\section{Some Extra Identification Results and Proof for Theorem \ref{thm:irreducible_model_id}}\label{sec:distinguish_appendix}
Let me first give a definition for reducibility and irreducibility for a general moment function $r$.

\begin{definition}[Reduced model, reducibility, and irreducibility]
\label{def:reduced_model_general}
Fix a model $(\Gamma,r)$ and a parameter value $\theta\in\Theta$, and let $\Theta_I(F;\Gamma,r)$ denote its identified set.
Partition $r(\cdot, \cdot; \theta)$ as
\begin{equation*}
r(u,z;\theta) = 
\begin{pmatrix}
r_1(z;\theta)\\
r_2(u,z;\theta)
\end{pmatrix}
\end{equation*}
where $r_1$ collects the components of $r$ that does not depend on $u$. For any $\lambda\in \real^{\dim(r_2)}$, $\gamma_2(\lambda,z;\theta) \equiv \sup\{\lambda^\top r_2(u,z;\theta):\ (u,z)\in \Gamma(\theta)\}$.

A model $(\tilde{\Gamma},\tilde{r})$ is a
\emph{reduced model} of $(\Gamma,r)$ at $\theta$ if the following two conditions hold:
\begin{itemize}
    \item \emph{(Construction):} There exist $\dim(r_2)$ linearly independent vectors $\lambda_1, \ldots, \lambda_{\dim(r_2)} \in \mathbb{R}^{\dim(r_2)}$ such that
    \begin{align*}
    \tilde{\Gamma}(\theta) &\equiv \left\{ (u,z) \in \Gamma(\theta) : \lambda_1^\top r_2(u,z;\theta) = \gamma_2(\lambda_1, z; \theta) \right\}, \\
    \tilde{r}(u,z;\theta) &\equiv 
\begin{pmatrix}
r_1(z;\theta)\\
 \gamma_2(\lambda_1, z; \theta)\\
\lambda_2^\top r_2(u,z;\theta)\\
\vdots\\
\lambda_{\dim(r_2)}^\top r_2(u,z;\theta)
\end{pmatrix}
    \end{align*}
    
    \item \emph{(Identification equivalence):} For every $F \in \mathcal{F}$, $\theta \in \Theta_I(F; \Gamma, r)$ if and only if $\theta \in \Theta_I(F; \tilde{\Gamma}, \tilde{r})$.
\end{itemize}
Model $(\Gamma,r)$ is \emph{reducible} if it admits at least one reduced model at some $\theta\in\Theta$, and it is
\emph{irreducible} otherwise.
\end{definition}

This definition reduces back to Definition \ref{def:reduced_model} when all components of $r$ depend on $u$. The proof for Theorem \ref{thm:irreducible_model_id} builds on the following two results. 

\begin{theorem}\label{thm:nondistinguishable}
Suppose Assumption \ref{assu:reg} holds and $\mathcal{F}$ is a convex set.
Let $\theta$ be an arbitrary parameter in $\Theta$ and partition $r(u,z;\theta) = (r_1(z;\theta), r_2(u,z;\theta))$. Recall that  $\gamma_2(\lambda,z;\theta) \equiv \sup\{\lambda^\top r_2(u,z;\theta):\ (u,z)\in \Gamma(\theta)\}$ and $\mathcal{S}_2 = \{\lambda\in \real^{\dim(r_2)}: \norm{\lambda} = 1 \}$.

If there exists some $F^*\in \mathcal{F}$ such that 
  \begin{equation}\label{eq:strict_center}
    \expt_{F^*} [r_1(Z;\theta) ] = 0 \text{ and }    \inf_{\lambda\in \mathcal{S}_2} \expt_{F^*} \gamma_2(\lambda,Z;\theta) > 0,
\end{equation}
then it is impossible to distinguish $\theta\in \Theta_I$ from $\theta\in \overline{\Theta}_I$ in finite samples. Note that the $\inf_{\lambda\in \mathcal{S}_2} \expt_{F^*} \gamma_2(\lambda,Z;\theta) > 0$ in \eqref{eq:strict_center} includes the case that $\inf_{\lambda\in \mathcal{S}_2} \expt_{F^*} \gamma_2(\lambda,Z;\theta) = +\infty$.
\end{theorem}

In the following, we say inequality \eqref{eq:strict_center} fails to hold for all $F\in \mathcal{F}$ if there does not exist an $F^*\in \mathcal{F}$ such that  \eqref{eq:strict_center} holds.

\begin{theorem}\label{lem:reducible_lemma}
Suppose Assumption \ref{assu:reg} hold and $\mathcal{F}$ is convex. 
Let $\theta$ be an arbitrary parameter in $\Theta$. Partition $r(u,z;\theta) = (r_1(z;\theta), r_2(u,z;\theta)$ and recall that  $\gamma_2(\lambda,z;\theta) \equiv \sup\{\lambda^\top r_2(u,z;\theta):\ (u,z)\in \Gamma(\theta)\}$ and $\mathcal{S}_2 = \{\lambda\in \real^{\dim(r_2)}: \norm{\lambda} = 1 \}$.

Suppose $\mathcal{F}'_\theta$ is nonempty, and that
the inequality \eqref{eq:strict_center} fails to hold for all $F\in \mathcal{F}$.  Then,
\begin{enumerate}
    \item there exists some $\tilde{\lambda}\in \mathcal{S}_2$ such that $\expt_F \gamma_2(\tilde{\lambda},Z;\theta)
      = 0$ for all $F\in \mathcal{F}'_\theta$.
    \item model $(\Gamma, r)$ is reducible at $\theta$. In particular, fix any $\lambda_2, ...,
      \lambda_{\dim(r_2)}$ such that $\tilde{\lambda}, \lambda_2, ..., \lambda_{\dim(r_2)}$ are linearly independent, and
      define reduced model $(\tilde{\Gamma}, \tilde{r})$ as
  \begin{eqnarray*}
  \tilde{\Gamma}(\theta)
  &=& \left\{(u,z)\in \Gamma(\theta): u\in  \argmax_{u \in \Gamma(z;\theta)}\tilde{\lambda}'r_2(u,z;\theta)
  \right\},\\
    \tilde{r}(u,z;\theta) & = & 
    \begin{pmatrix}
      r_1(Z;\theta)\\
\gamma_2(\tilde{\lambda}, z;\theta) \\
 \lambda_2'r_2(u,z;\theta) \\   
 \lambda_3'r_2(u,z;\theta) \\   
 \vdots \\
\lambda'_{\dim(r_2)} r_2(u,z;\theta)
    \end{pmatrix}.
  \end{eqnarray*}
   Then, for any $F\in \mathcal{F}$, $\theta\in \Theta_I(F;\Gamma, r)$
  if and only if $\theta\in \Theta_I(F; \tilde{\Gamma}, \tilde{\gamma})$. 
  \end{enumerate}
\end{theorem}

In the following, I first prove Theorem \ref{thm:irreducible_model_id} with Theorems \ref{thm:nondistinguishable} and 
\ref{lem:reducible_lemma}. After that, I would prove Theorems \ref{thm:nondistinguishable} and \ref{lem:reducible_lemma} in two seperate subsections.

\begin{proof}[Proof for Theorem \ref{thm:irreducible_model_id}]
suppose the model $(\Gamma, r)$ is irreducible at $\theta$.  Consider the following cases:
\begin{itemize}
  \item When $\mathcal{F}'_\theta$ is empty, $\mathcal{F}_\theta$ is also empty so that both
$\theta\in \Theta_I(F)$ and $\theta\in \overline{\Theta}_I(F)$ are false for any $F\in \mathcal{F}$, which implies that
$\theta\in \overline{\Theta}_I$ and $\theta\in \Theta_I$ cannot be distinguished in finite samples.
  \item When $\mathcal{F}'_\theta$ is nonempty, Theorem \ref{lem:reducible_lemma} implies that there exists some $F^*\in
    \mathcal{F}$ which satisfies \eqref{eq:strict_center} (otherwise the model would be reducible at $\theta$).  Theorem \ref{thm:nondistinguishable} then implies that
    $\theta\in \overline{\Theta}_I$ and $\theta\in \Theta_I$ cannot be distinguished in finite samples.
\end{itemize}
Since $\theta\in \overline{\Theta}_I$ and $\theta\in \Theta_I$ are indistinguishable in both cases, the proof is now
complete.
\end{proof}

\subsection{Proof for Theorem \ref{thm:nondistinguishable}}
The proof for Theorem \ref{thm:nondistinguishable} builds on the following lemma which I would prove after proving Theorem \ref{thm:nondistinguishable}.

\begin{lemma}\label{lem:bound_for_difference}
  Suppose Assumption \ref{assu:reg} holds.  Let $F$ be an arbitrary element in $\mathcal{F}$ and 
  let $\theta$ be an arbitrary parameter in $\Theta$.  Partition
  $r(u,z;\theta) = (r_1(z;\theta), r_2(u,z;\theta))$. 
  \begin{enumerate}
    \item\label{enu:lemma1_core} if $\expt_F[r_1(Z;\theta)] = 0$ and $\inf_{\lambda\in \mathcal{S}_2} \expt_F \gamma_2(\lambda, Z;\theta) > 0$, then $\theta\in \Theta_I(F)$. (Note that the $\inf_{\lambda\in \mathcal{S}_2} \expt_F \gamma_2(\lambda, Z;\theta) > 0$ includes the case that $\inf_{\lambda\in \mathcal{S}_2} \expt_F \gamma_2(\lambda, Z;\theta) =+\infty$.)
    \item\label{enu:lemma1_side} if $\theta\in \overline{\Theta}_I(F) \backslash \Theta_I(F)$, then 
  \begin{equation*}
    \expt_F[r_1(Z;\theta)] = 0\text{ and }\inf_{\lambda\in \mathcal{S}_2}\expt_F \gamma_2(\lambda, Z;\theta) = 0,
  \end{equation*}
  or equivalently, $\inf_{\lambda\in \mathcal{S}} \expt_F \gamma(\lambda,Z;\theta) = 0.$
  \end{enumerate}
\end{lemma}

\begin{proof}[Proof for Theorem \ref{thm:nondistinguishable}]
Define set $\mathcal{F}_\theta^*$ as
\begin{equation*}
  \mathcal{F}_\theta^*\coloneqq \left\{ F\in \mathcal{F}:   \expt_{F} [r_1(Z;\theta) ] = 0 \text{ and }    
  \inf_{\lambda\in \mathcal{S}_2} \expt_{F} \gamma_2(\lambda,Z;\theta) > 0. \right\}
\end{equation*}
Since \eqref{eq:strict_center} holds, we know $\mathcal{F}_\theta^*$ is nonempty. By Lemma \ref{lem:bound_for_difference}, 
$\mathcal{F}_\theta^*\subseteq \mathcal{F}_\theta \subseteq \mathcal{F}'_\theta$. 
Hence, both $\sup_{F\in \mathcal{F}_\theta}\expt_F \phi_n
$ and $\sup_{F\in \mathcal{F}'_\theta}\expt_F \phi_n$ are well defined and finite. 
$\mathcal{F}_\theta^*\subseteq \mathcal{F}_\theta \subseteq \mathcal{F}'_\theta$ also implies that 
$\sup_{F\in \mathcal{F}^*_\theta} \expt_F\phi_n  \le \sup_{F\in \mathcal{F}_\theta} \expt_F\phi_n \le \sup_{F\in
\mathcal{F}'_\theta} \expt_F\phi_n$.
Therefore, to show the desired result, we only need to show that for any $F \in
\mathcal{F}'_\theta$, $\expt_{F}\phi_n \le \sup_{F\in \mathcal{F}^*_\theta} \expt_F\phi_n$. 

For each $F\in \mathcal{F}$, define  Define $\psi(F) = \inf_{\lambda\in \mathcal{S}_2}\expt_F\gamma_2(\lambda,Z;\theta)$.
For any $F_1, F_2\in \mathcal{F}$, let $F_\delta = \delta F_1 + (1-\delta)F_2$ for any $\delta\in [0, 1]$. Then, 
\begin{eqnarray*}
  \psi(F_\delta) & = & \inf_{\lambda\in \mathcal{S}_2} \Big(\delta \expt_{F_1}\gamma_2(\lambda,Z;\theta) + (1-\delta)
  \expt_{F_2}\gamma_2(\lambda,Z;\theta)\Big) \\
                 & \ge &  \inf_{\lambda\in \mathcal{S}_2}\delta\expt_{F_1}\gamma_2(\lambda,Z;\theta) + \inf_{\lambda\in
                 \mathcal{S}_2} (1-\delta)\expt_{F_2}\gamma_2(\lambda,Z;\theta) \\
                 & = & \delta \psi(F_1) + (1-\delta) \psi(F_2)
\end{eqnarray*}
Therefore, $\psi$ is a concave function. 

Now, fix an arbitrary $F\in \mathcal{F}'_\theta$. For any $F^*\in \mathcal{F}^*$ and any $k\ge 1$, define $F_k\coloneqq
(1-\frac{1}{k})F + \frac{1}{k}F^*$. Since $F\in \mathcal{F}'_\theta$, $\expt_F[r_1(Z;\theta)] = 0$ and
$\inf_{\lambda\in \mathcal{S}_2} \expt_{F} \gamma_2(\lambda,Z;\theta) \ge 0$. Therefore, the concavity of $\psi$ implies
that $F_k\in \mathcal{F}^*_\theta$ for all $k\ge 1$.
 Since $\expt_F \phi_n = \lim_{k\to\infty} \expt_{F_k} \phi_n$, we know 
\begin{equation*}
\expt_F \phi_n \le \sup_{k\ge 1}\expt_{F_k} \phi_n \le \sup_{F' \in \mathcal{F}^*} \expt_{F'} \phi_n.
\end{equation*}
This completes the proof.
\end{proof}

\subsubsection{Proof for Lemma \ref{lem:bound_for_difference}}
The proof of this result builds on the following lemmas, whose proofs will be presented later.
\begin{lemmasec}\label{lem:support_func_interior}
  Suppose set $A$ is a nonempty closed convex set in $\real^d$. Let $\interior{A}$ denote the interior of $A$. Then, 
  $x\in\interior{A}$ if and only if
  \begin{equation}\label{eq:inf_supp_pos}
    \inf_{\lambda\in \real^d: \norm{\lambda} = 1} \sup\{\lambda'(t - x): t\in A\} > 0. 
  \end{equation}
  Note that \eqref{eq:inf_supp_pos} includes the case that $\inf_{\lambda\in \real^d: \norm{\lambda} = 1} \sup\{\lambda'(t - x): t\in A\} =+\infty$ which could happen when $A=\real^d$. 
\end{lemmasec}

\begin{lemmasec}\label{lem:convex_set_interior_perturbation}
  Suppose set $A$ is a nonempty set in $\real^d$. Suppose $x\in\interior (\co A)$, there there exists 
  some $\epsilon > 0$, a positive integer $K > 0$ and $a_1,....,a_K\in A$, such that, $x\in
  \interior(\co\{a'_1,...,a'_K\})$ for any $a'_1,...,a'_K$ with $\norm{a_i - a'_i} < \epsilon$ for $i=1,...,K$. 
\end{lemmasec}

\begin{lemmasec}\label{lem:suff_id_interior}
  Suppose Assumption \ref{assu:reg} hold. Then, for any $F\in \mathcal{F}$,  $0 \in \interior(\clco \expt_F 
  \cl\mathcal{R}(Z;\theta))$ implies $\theta \in \Theta_I(F)$.
\end{lemmasec}

\begin{proof}[Proof of Lemma \ref{lem:bound_for_difference}]
  Fix an arbitrary $F\in \mathcal{F}$. In the following of the proof, I will abbreviate $\expt_F$ as $\expt$, and
  $\Theta_I(F)$ as $\Theta_I$. Recall also that $\expt_{I}$ stands for the Aumann integral. 
  Lemma \ref{lem:suff_random_set} implies that $\cl\mathcal{R}(\cdot;\theta)$ is an integrable random closed set.
  The proof will be conducted in three steps. 

  \noindent {\bf Step 1:} Lemma \ref{lem:bound_for_difference}\ref{enu:lemma1_core} holds when $\dim(r_1) = 0$. 

  Suppose $\dim(r_1) = 0$. I need to prove that  $\inf_{\lambda\in \mathcal{S}}\expt_F \gamma(\lambda, Z;\theta) > 0$
  implies $\theta\in \Theta_I$ in this step. Suppose $\inf_{\lambda\in \mathcal{S}}\expt_F \gamma(\lambda, Z;\theta) > 0$.
  I'm going to prove $\theta\in \Theta_I$ by contradiction.

  Suppose, for the purpose of contradiction, that $\theta\notin \Theta_I$. Then, Lemma \ref{lem:support_func_interior}
  and \ref{lem:suff_id_interior} implies that 
  \begin{displaymath}
    \inf_{\lambda\in \mathcal{S}} \sup\{\lambda't: t\in \clco \expt
    \cl\mathcal{R}(Z;\theta)\} \le 0.
  \end{displaymath}
  By Lemma \ref{lem:random_set_all}\ref{enu:wa4}, and the fact that
  $\clco\mathcal{R}(Z;\theta)\subseteq\clco\cl\mathcal{R}(Z;\theta)$, and that
  $\expt_I\clco\mathcal{R}(Z;\theta)\subseteq \expt\clco\mathcal{R}(Z;\theta)$, we
  know
  \begin{equation}
    \label{eq:temp_q23i}
    \inf_{\lambda\in \mathcal{S}} \sup\{\lambda't: t\in \expt_I
    \clco\mathcal{R}(Z;\theta)\} \le 0
  \end{equation}
  Since $\sup\{\lambda't: t\in \expt_I \clco\mathcal{R}(Z;\theta)\}$ is a lower semi-continuous function of $\lambda$ and
  $\mathcal{S}$ is compact, there exists some $\tilde{\lambda}$ such that $\sup\{\tilde{\lambda}'t: t\in \expt_I
  \clco\mathcal{R}(Z;\theta)\} \le 0$. Note that
  \begin{equation}
    \label{eq:niqwdema}
    \sup\{\tilde{\lambda}'t: t\in \expt_I
  \clco\mathcal{R}(Z;\theta)\} = -\inf_{f\in S^1(\clco\mathcal{R}(Z;\theta))} \expt [-\tilde{\lambda}'f]
  \end{equation}
  where $S^1$ is defined in Definition \ref{def:int_sel_awofei}. Apply Lemma
  \ref{lem:random_set_all}\ref{enu:wa3} with $\zeta(t) = -\lambda't$ to get
  \begin{eqnarray}
    && -\inf_{f\in S^1(\clco\mathcal{R}(Z;\theta))} \expt [-\tilde{\lambda}'f] \nonumber\\
    & = & -\expt \inf\{-\tilde{\lambda}'t: t \in \clco\mathcal{R}(Z;\theta)\} \nonumber\\
    & = & \expt \sup\{\tilde{\lambda}'t: t \in \clco\mathcal{R}(Z;\theta)\}. \label{eq:fnmwdzwa}
  \end{eqnarray}
  Equation \eqref{eq:niqwdema} and \eqref{eq:fnmwdzwa} imply
  \begin{equation}
    \label{eq:nfwiamso}
    \expt \sup\{\tilde{\lambda}'t: t \in \clco\mathcal{R}(Z;\theta)\} =
    \sup\{\tilde{\lambda}'t: t\in \expt_I
    \clco\mathcal{R}(Z;\theta)\}  \le 0.
  \end{equation}
  In addition, since $\mathcal{R}(z;\theta)$ is a subset of the Euclidean space,
  \begin{equation}
    \label{eq:fjei2any}
    \sup\{\tilde{\lambda}'t: t\in \clco \mathcal{R}(z;\theta)\} = \sup\{\tilde{\lambda}'t: t\in \mathcal{R}(z;\theta)\}.
  \end{equation}
 Equation \eqref{eq:nfwiamso} and \eqref{eq:fjei2any} then imply
  \begin{displaymath}
    \inf_{\lambda\in \real^{\dim(r)}} \expt \sup\{\lambda't: t\in \mathcal{R}(Z;\theta)
    \} \le 0.
  \end{displaymath}
  This contradicts the fact that $\theta$ satisfy $\inf_{\lambda\in \mathcal{S}}\expt_F \gamma(\lambda, Z;\theta) > 0$.

  \noindent {\bf Step 2:} Lemma \ref{lem:bound_for_difference}\ref{enu:lemma1_core} holds when $\dim(r_1) > 0$. 

  Recall that $\mathcal{H}(\theta, F)$ is defined as the set of all joint
  distributions $H$ for $(U,Z)$ which satisfy that $\prob_H[(U,Z)\in \Gamma(\theta)] = 1$ and that $H$'s marginal
  distribution for $Z$ equals $F$.
  \begin{itemize}
    \item   When $\dim(r_2) = 0$, for any $H \in \mathcal{H}(\theta, F)$, we have $\expt_H [r(U,Z;\theta)] = \expt_F [r_1(Z;\theta)]$.  Therefore, \eqref{eq:id_def_min} is equivalent to $\expt_F[r_1(Z;\theta)] = 0$.
  Hence, $\expt_F r_1(Z;\theta) = 0$ implies $\theta\in \Theta_I(F;\Gamma, r_1) = \Theta_I(F;\Gamma, r)$ 
  by Definition \ref{def:id_set}.
    \item When $\dim(r_2) > 0$, note that \eqref{eq:id_def_min} is equivalent to the following condition:
  \begin{equation*}
    \expt_F[r_1(Z;\theta)]= 0 \text{ and } \inf_{H\in \mathcal{H}(\theta, F)} \norm{\expt_H [r_2(U,Z;\theta)]} = 0.
  \end{equation*}
  which implies that $\Theta_I(F;\Gamma, r) = \Theta_I(F;\Gamma, r_1) \cap \Theta_I(F;\Gamma, r_2)$ by Definition
  \ref{def:id_set}. Following the same proof in the previous paragraph, we know that $\expt_F[r_1(Z;\theta)] = 0$
  implies $\theta\in \Theta_I(F;\Gamma, r_1)$. Following the same proof in Step 1, we know that $\inf_{\lambda\in
  \mathcal{S}_2}\expt_F \gamma_2(\lambda, Z;\theta) > 0$ implies $\theta\in \Theta_I(F;\Gamma, r_2)$. As a result, 
$\expt_F[r_1(Z;\theta)] = 0$ and $\inf_{\lambda\in \mathcal{S}_2} \expt_F \gamma_2(\lambda, Z;\theta) > 0$ implies
$\theta\in \Theta_I(F;\Gamma, r)$.
  \end{itemize}
   Step 1 and 2 completes the proof for Lemma \ref{lem:bound_for_difference}\ref{enu:lemma1_core}.

   \noindent {\bf Step 3:} Lemma \ref{lem:bound_for_difference}\ref{enu:lemma1_side} holds.

   Suppose $\theta\in \overline{\Theta}_I(F) \backslash \Theta_I(F)$. Because $\theta\in  \overline{\Theta}_I(F) $, we know 
   \begin{equation*}
  \expt_F[r_1(Z;\theta)] = 0\text{ and }\inf_{\lambda\in \mathcal{S}_2} \expt_F \gamma_2(\lambda, Z;\theta) \ge 0.
   \end{equation*}
   Moreover, $\theta\notin \Theta_I(F)$ implies that there is no 
   \begin{equation*}
  \expt_F[r_1(Z;\theta)] = 0\text{ and }\inf_{\lambda\in \mathcal{S}_2} \expt_F \gamma_2(\lambda, Z;\theta) > 0.
   \end{equation*}
   Hence, we must have $\expt_F[r_1(Z;\theta)] = 0$ and $\inf_{\lambda\in \mathcal{S}_2} \expt_F \gamma_2(\lambda, Z;\theta) = 0.$
\end{proof}

\begin{proof}[Proof of Lemma \ref{lem:support_func_interior}]
  By Theorem 2.2.3 (on page 138) in \cite{hiriart-urruty_fundamentals_2001}, we know that $x\in \interior A$ if and only if for any $\lambda$ with $\norm{\lambda} = 1$, $\sup\{\lambda'(t - x): t\in A\} > 0$. Therefore, I only need to show that $\sup\{\lambda'(t - x): t\in A\} > 0$ for any $\lambda$ with $\norm{\lambda} = 1$ if and only if 
  \begin{equation*}
    \inf_{\lambda\in \real^d: \norm{\lambda} = 1} \sup\{\lambda'(t - x): t\in A\} > 0. 
  \end{equation*}
  The "if" part of this claim follows from the definition of $\inf$. To show the "only if" part of this claim, note that $\sup\{\lambda'(t - x): t\in A\}$ is a lower semi-continuous function of $\lambda$ and that $\{\lambda \in \real^d: \norm{\lambda}=1\}$ is a compact set. Note also that there must be $\inf_{\lambda\in \real^d: \norm{\lambda} = 1} \sup\{\lambda'(t - x): t\in A\} >= 0$.  Therefore, if $\inf_{\lambda\in \real^d: \norm{\lambda} = 1} \sup\{\lambda'(t - x): t\in A\} < +\infty$, this infimum is achieved by some $\lambda$ with $\norm{\lambda} = 1$ so that $\inf_{\lambda\in \real^d: \norm{\lambda} = 1} \sup\{\lambda'(t - x): t\in A\} > 0$. If $\inf_{\lambda\in \real^d: \norm{\lambda} = 1} \sup\{\lambda'(t - x): t\in A\} = +\infty$, then we automatically have $\inf_{\lambda\in \real^d: \norm{\lambda} = 1} \sup\{\lambda'(t - x): t\in A\} > 0.$ 
\end{proof}

\begin{proof}[Proof of Lemma \ref{lem:convex_set_interior_perturbation}]
  By \cite{gustin_interior_1947}, there exists some $a_1,...,a_K\in A$ such that $x$ is in the interior of $\co
  \{a_1,...,a_K\}$ and $K \le 2 d$. By Lemma \ref{lem:support_func_interior}, we know that $\lambda'(a_i - x) > 0$ for
  all $\lambda$ with $\norm{\lambda} = 1$ and for all $i=1,...,K$. Therefore, there exists some $\epsilon > 0$ such that
  for any $i=1,...,K$ and for any $a'_i$ with $\norm{a'_i - a_i} < \epsilon$, we have $\lambda'(a'_i - x) > 0$. By Lemma
  \ref{lem:support_func_interior},  this is equivalent to that $x\in \interior(\co\{a'_1,...,a'_K\})$ for any
  $a'_1,...,a'_K$ with $\norm{a_i - a'_i} < \epsilon$ for $i=1,...,K$. 
\end{proof}

\begin{proof}[Proof of Lemma \ref{lem:suff_id_interior}]
  Fix an arbitrary $F\in \mathcal{F}$. In the following, I will abbreviate $\expt_F$ as $\expt$, and $\Theta_I(F)$ as
  $\Theta_I$. Recall that $\expt_{I}$ stands for the Aumann integral. Recall also that the probability space
  $(\mathcal{Z},\mathscr{Z}, F)$ denotes the completion of Borel probability space $(\mathcal{Z},\mathscr{B}_Z, F)$.

  Under Assumption \ref{assu:reg}, $\cl\mathcal{R}(Z;\theta)$ 
  is an integrable random closed set in $(\mathcal{Z}, \mathscr{Z}, F_Z)$. 
  Suppose $0 \in \interior(\clco \expt \cl\mathcal{R}(Z;\theta))$ is true, I want
  to prove that $\theta\in \Theta_I$. 

  Because $\clco A = \clco\cl A$ for any subset $A$ in an Euclidean space, and because $\expt \cl
  \mathcal{R}(Z;\theta) = \cl(\expt_I \cl\mathcal{R}(Z;\theta))$ by Definition
  \ref{def:integrable_bound}, $0\in \interior(\clco \expt \cl \mathcal{R}(Z;\theta))$
  imply $0 \in \interior\left( \clco \expt_I \cl \mathcal{R}(Z;\theta) \right)$. 
  Furthermore, because Proposition 2.1.8 in \cite{hiriart-urruty_fundamentals_2001} implies that   $\interior(\clco A)
  = \interior(\co A)$ for any subset $A$ in an Euclidean space, we know $0 \in \interior\left( \clco \expt_I \cl
  \mathcal{R}(Z;\theta) \right)$ implies that $0 \in \interior \left( \co \expt_I \cl \mathcal{R}(Z;\theta) \right)$.
  By Lemma \ref{lem:convex_set_interior_perturbation}, we know there exists some $\epsilon > 0$, some positive integer
  $K$ and some $v_1,...,v_{K}\in \expt_I \cl \mathcal{R}(Z;\theta)$ such that $0 \in \interior
  (\co\{\tilde{v}_1,...,\tilde{v}_K\})$ for any $(\tilde{v}_1,...,\tilde{v}_K)$ with $\norm{\tilde{v}_i - v_i} < \epsilon$
  for any $i=1,...,K$. 

  For any $k = 1,...,K$.  Because $v_{k}\in \expt_I \cl \mathcal{R}(Z;\theta)$, there exists  $f_k\in
  S^1(\cl\mathcal{R}(Z;\theta))$ such that $v_k = \expt f_k(Z)$. Because every measurable function in
  $(\mathcal{Z},\mathscr{Z}, F_Z)$ can be well approximated by a Borel measurable function, there exists some Borel function
  $\tilde{f}_k$ such that $\prob(f_k(Z) = \tilde{f}_k(Z)) = 1$. Therefore, we know $\expt \tilde{f}_k(Z) = v_k$ and
  \begin{equation*}
    \prob\left( \inf_{u\in \Gamma(Z;\theta)} \norm{r(u,Z;\theta) - \tilde{f}_k(Z)} = 0 \right) = 1.
  \end{equation*}
  Since $\{(z,u):u\in \Gamma(z;\theta)\}\times \real^{\dim(r)}$ is a Borel set, and
  that $(u,z)\mapsto \norm{r(u,z;\theta) - \tilde{f}_k(z)}$ is a Borel measurable function, 
  Lemma \ref{lem:selection} in Appendix \ref{sec:selection} implies that there exists a universally
  measurable function $g:\mathcal{Z}\mapsto\mathcal{U}$, such
  that for almost every $z\in \mathcal{Z}$, $g_k(z) \in \Gamma(z;\theta)$ and
  \begin{displaymath}
    \norm{r\big(g_k(z), z\big) - \tilde{f}_k(z)} \le \epsilon + \inf_{u\in \Gamma(z;\theta)} \norm{r(u,z;\theta) - \tilde{f}_k(z)}.
  \end{displaymath}
  By the construction of $g_k$, $\norm{v_k - \expt r(g_k(Z), Z)} < \epsilon$. 

  As a result, I have shown that there exists function $g_1,...,g_K$ in $(\mathcal{Z},\mathscr{Z}, F_Z)$ such that
  $\prob(g_k(Z) \in \Gamma(Z;\theta)) = 1$ for each $k=1,...,K$ and $0 \in \co\{\expt r(g_1(Z),Z),..., \expt
  r(g_K(Z),Z)\}$. This implies that there exists a joint distribution $H$ for $(U, Z)$ such that (i) $H$'s marginal
  distribution for $Z$ is $F_Z$, (ii) $\prob_H\left( U \in \Gamma(Z;\theta) \right) = 1$ and (iii) $\expt_H r(U,
  Z;\theta) = 0$. Hence, $\theta\in \Theta_I$.
\end{proof}

\subsection{Proof for Theorem \ref{lem:reducible_lemma}}
  Fix $\theta$ to be an arbitrary parameter with which \eqref{eq:strict_center} does not hold for all $F\in \mathcal{F}$.
  The proof will be divided into two parts: Part 1 deals with the first part of the result and Part 2 deals with the second
  part of the result.

  \noindent{\bf Part 1}
  First of all, the fact that $\mathcal{F}'_\theta$ is nonempty and \eqref{eq:strict_center} fails to hold for all $F\in
  \mathcal{F}$ implies that $\dim(r_2) > 0$. For any $F\in \mathcal{F}'_\theta$, Theorem \ref{thm:weak_sharp_support_function}
  implies that $\inf_{\lambda\in \mathcal{S}} \expt_{F} \gamma(\lambda,Z;\theta) \ge 0$, which is equivalent to 
  \begin{equation*}
    \expt_F [r_1(Z;\theta)] = 0\text{ and }\inf_{\lambda\in \mathcal{S}_2} \expt_{F} \gamma_2(\lambda,Z;\theta) \ge 0.
  \end{equation*}
  Because \eqref{eq:strict_center} fails to hold for all $F\in \mathcal{F}$, we know that for any $F\in \mathcal{F}'_\theta$, 
  $\inf_{\lambda\in \mathcal{S}_2} \expt_{F} \gamma_2(\lambda,Z;\theta) = 0$. Since $\gamma_2(\lambda, Z;\theta)$ is
  lower semi-continuous in $\lambda$ and $\mathcal{S}_2$ is a compact set, we know that for each $F\in
  \mathcal{F}'_\theta$, there exists some $\lambda\in \mathcal{S}_2$ 
  such that $\expt_F \gamma_2(\lambda,Z;\theta) = 0$. For each $F\in
  \mathcal{F}'_\theta$, define $\Lambda(F) \coloneqq \{\lambda\in \mathcal{S}_2: \expt_F \gamma_2(\lambda, Z;\theta) = 0 \}$.
  Then, for each $F\in \mathcal{F}'_\theta$, $\Lambda(F)$ is nonempty.  To show the first result of Theorem
  \ref{lem:reducible_lemma}, I only need to show that there exists some $F^*\in \mathcal{F}'_\theta$ such that 
  $\cap_{F\in \mathcal{F}'_\theta}\Lambda(F) = \Lambda(F^*)$.  When $\mathcal{F}'_\theta$ only contains one element
  $F^*$, it's trivially true that $\cap_{F\in \mathcal{F}'_\theta}\Lambda(F) = \Lambda(F^*)$. So, I suppose
  $\mathcal{F}'_\theta$ contains at least two elements in the remaining of the proof in this part.

  Note that $\mathcal{F}'_\theta$ is a convex set because
  $\mathcal{F}$ is convex and $\mathcal{F}'_\theta = \{F\in \mathcal{F}: \expt_F \gamma(\lambda,Z;\theta) \ge 0, \forall
  \lambda\in \mathcal{S}\}$.
  The relative interior $\ri \mathcal{F}'_\theta$ defined as $\ri\mathcal{F}'_\theta \coloneqq \{F\in
    \mathcal{F}'_\theta: \forall F'\in \mathcal{F}'_\theta,\ \exists \delta > 1 \text{ such that } \delta F + (1-\delta) F'\in
  \mathcal{F}'_\theta\}$ should contain at least two elements because $\mathcal{F}'_\theta$ contains at least two
  elements. 

  To proceed, I claim that for any $F_1, F_2\in \mathcal{F}'_\theta$ and any $\delta\in (0, 1)$, $\Lambda(F_1) \cap
  \Lambda(F_2) = \Lambda(F_\delta)$ where $F_\delta\coloneqq \delta F_1 + (1-\delta) F_2$. To see why this is true, note
  that for any $\lambda \in \Lambda(F_1) \cap \Lambda(F_2)$, $\expt_{F_\delta} \gamma_2(\lambda,
  Z;\theta) = \delta \expt_{F_1}\gamma_2(\lambda,Z;\theta) + (1-\delta)\expt_{F_2}\gamma_2(\lambda, Z;\theta) = 0$. Hence,
  $\Lambda(F_\delta) \supseteq \Lambda(F_1) \cap \Lambda(F_2)$. Now, for any $\lambda\in \mathcal{S}_2\backslash
  (\Lambda(F_1) \cap \Lambda(F_2))$, we know the following is true:
  \begin{itemize}
    \item $\expt_{F_1}\gamma_2(\lambda, Z;\theta) \ge 0$, because $F_1\in \mathcal{F}'_\theta$;
    \item $\expt_{F_2}\gamma_2(\lambda, Z;\theta) \ge 0$, because $F_1\in \mathcal{F}'_\theta$;
    \item either $\expt_{F_1}\gamma_2(\lambda, Z;\theta) > 0$ or $\expt_{F_2}\gamma_2(\lambda,
  Z;\theta) > 0$, because $\lambda\notin \Lambda(F_1) \cap \Lambda(F_2)$.
  \end{itemize}
  Therefore, $\expt_{F_\delta}\gamma_2(\lambda,
  Z;\theta) = \delta \expt_{F_1}\gamma_2(\lambda,Z;\theta) + (1-\delta)\expt_{F_2}\gamma_2(\lambda, Z;\theta) > 0$. Hence,
  for any  $\lambda\in \mathcal{S}_2\backslash (\Lambda(F_1) \cap \Lambda(F_2))$, 
  $\lambda\notin \Lambda(F_\delta)$.  Hence, $\Lambda(F_\delta) \subseteq \Lambda(F_1) \cap \Lambda(F_2)$. Combine both
  results, I conclude that $\Lambda(F_1) \cap \Lambda(F_2) = \Lambda(F_\delta)$ for any $\delta\in (0, 1)$. 

  Next, I claim that for any two $F_1$, $F_2$ in $\ri \mathcal{F}'_\theta$, $\Lambda(F_1) = \Lambda(F_2)$. To
  see why this is true, note that by the definition of $\ri \mathcal{F}'_\theta$, there must exists $F_3$ and $F_4$ in
  $\mathcal{F}'_\theta$ and $\delta_1, \delta_2\in (0, 1)$ such that $F_1 = \delta_1 F_3 + (1-\delta_1) F_4$ and $F_2
  = \delta_2 F_3 + (1-\delta_2) F_4$. By the preceding result, we know $\Lambda(F_1) = \Lambda(F_2) = \Lambda(F_3) \cap
  \Lambda(F_4)$. 

  Finally, let $F^*$ be an arbitrary element in $\ri \mathcal{F}'_\theta$. 
  I claim that $\cap_{F\in \mathcal{F}'_\theta}\Lambda(F) = \Lambda(F^*)$. To see why this is true, note that for any
  $F\in \mathcal{F}'_\theta$, there must exists $F'\in \mathcal{F}'_\theta$ with $F'\neq F$ because
  $\mathcal{F}'_\theta$ is assumed to have at least two elements.  Because $\frac{1}{2}F + \frac{1}{2}F'\in
  \ri\mathcal{F}_\theta$, the claims which I proved in the above paragraphs implies that $\Lambda(F^*) = \Lambda(F)\cap
  \Lambda(F')$.  As a result, $\Lambda(F^*)\subseteq \Lambda(F)$ for all $F\in \mathcal{F}_\theta$. Hence, $\cap_{F\in
  \mathcal{F}'_\theta}\Lambda(F) = \Lambda(F^*)$.

  \noindent{\bf Part 2}
  I am going to prove the second part of the result in two steps. 

  \emph{Step 1: $\forall F\in \mathcal{F}$, $\theta\in \Theta_I(F;\tilde{\Gamma},
  \tilde{r})$ implies that $\theta\in \Theta_I(F;\Gamma, r)$}. To prove this result, suppose 
  $\theta\in \Theta_I(F;\tilde{\Gamma}, \tilde{r})$ for some $F\in \mathcal{F}$. Because of the definition of
  $\Theta_I(F;\tilde{\Gamma}, \tilde{r})$, there exists some joint distribution $H$ of $(U, Z)$ such that (i) $\prob_H(
  (U, Z)\in \tilde{\Gamma}(\theta)) = 1$; (ii) $\expt_H r_1(Z;\theta) = 0$, $\expt_H \gamma_2(\tilde{\lambda}, Z;\theta)
  = 0$ and $\expt_H \lambda'_i r_2(U,Z;\theta) = 0$ for
  $i=2,...,\dim(r_2)$; (iii) the marginal distribution of $H$ for $Z$ is $F$. 

  Because of the construction of $\tilde{\Gamma}$ in this lemma, and because $\prob_H( (U, Z)\in \tilde{\Gamma}(\theta))
  = 1$, we know that $\prob_H ( \lambda'_1 r(U,Z;\theta)= \gamma_2(\tilde{\lambda},Z;\theta)) = 1$. Therefore, in
  addition to $\expt_H \lambda'_i r_2(U,Z;\theta) = 0$ for each $i=2,...,\dim(r_2)$, we also have $\expt_H
  \tilde{\lambda}' r_2(U,Z;\theta) = 0$. Because
  $\tilde{\lambda}$, $\lambda_2$, ..., $\lambda_{\dim(r_2)}$ are linearly independent, this implies that $\expt_H
  r_2(U,Z;\theta) = 0\in \real^{\dim(r_2)}$. Moreover, since $\tilde{\Gamma}(\theta) \subseteq \Gamma(\theta)$,
  $\prob_{H}( (U,Z)\in \Gamma(\theta) ) = 1$. As a result, $\theta\in \Theta_I(F;\Gamma, r)$.

  \emph{Step 2: $\forall F\in \mathcal{F}$, $\theta\in \Theta_I(F;\Gamma,
  r)$ implies that $\theta\in \Theta_I(F; \tilde{\Gamma}, \tilde{r})$}. To prove this result, suppose $\theta\in
  \Theta_I(F; \Gamma, r)$ for some $F\in \mathcal{F}_\theta$. By the definition of $\Theta_I(F; \Gamma, r)$, there
  exists some joint distribution $H$ of $(U, Z)$ such that (i) $\prob_H((U, Z)\in \Gamma(\theta)) = 1$; (ii) $\expt_H
  r_1(Z;\theta) = 0$ and $\expt_H
  r_2(U,Z;\theta) = 0$; and (iii) the marginal distribution of $H$ for $Z$ is $F$. Note that since $F\in
  \mathcal{F}_\theta\subseteq \mathcal{F}'_\theta$, we know $\expt_F \gamma_2(\tilde{\lambda}, Z;\theta) = 0$ by the
  construction of $\tilde{\lambda}$.

  Define $\phi(u,z;\theta) = \gamma_2(\tilde{\lambda}, z;\theta) - \tilde{\lambda}'r_2(u,z;\theta)$.
  To show $\theta\in \Theta_I(F; \tilde{\Gamma}, \tilde{r})$, I only need to verify that $\prob_H(\phi(U,Z;\theta) = 0)
  = 1$. By the construction of $H$, there is $\expt_H \tilde{\lambda}'r_2(U,Z;\theta) = 0$. Moreover, by the construction
  of $\tilde{\lambda}$, there is $\expt_F \gamma_2(\tilde{\lambda}, Z;\theta) = 0$ which implies that $\expt_H
  \gamma_2(\tilde{\lambda}, Z;\theta) = 0$. Therefore, we have $\expt_H \phi(U,Z;\theta) = 0$. Recall
  $\gamma_2(\tilde{\lambda}, z;\theta) \coloneqq \sup_{u\in \Gamma(z;\theta)}\tilde{\lambda}'r_2(u,z;\theta)$. Because 
  $\prob_H((U, Z)\in \Gamma(\theta)) = 1$, there is $\prob_H(\phi(U,Z;\theta) \ge 0) = 1$. Combine this result with
  $\expt_H \phi(U,Z;\theta) = 0$, it must be true that $\prob_H(\phi(U,Z;\theta) = 0) = 1$. This proves the desired
  result.

\section{Proof for the claim in the Example that illustrates Scenario (\emph{i})}\label{sec:proof_claim_example}
Let $\Theta_I(F)$ be the identified set for models with support restriction \eqref{eq:entry_game_support} and moment restriction \eqref{eq:entry_medium}. Let $\tilde{\Theta}(F)$ be the corresponding support-function set. Let $\Theta_I'(F)$ be the identified set for the augmented model with counterfactual support restriction \eqref{eq:entry_counterfactual_support} and counterfactual moment $\expt[\sum_{j\in \{0, 1\}} \tilde{\pi}_j - \tilde{\theta}] = 0$ with $\tilde{\pi}_j = \tilde{Y}_j[\tilde{X}_j^\top \alpha - \Delta_j \tilde{Y}_{1-j} + U_j]$ and $\tilde{X} = \phi(X)$. Let $\tilde{\Theta}'(F)$ be the corresponding support-function set for the augmented model. Let $\Xi(F)$ be the projection of $\Theta_I'(F)$ onto the $\tilde{\theta}$-coordinate, and $\tilde{\Xi}(F)$ be the projection of $\tilde{\Theta}'(F)$ onto the $\tilde{\theta}$-coordinate.

At the end of this section, we will prove the following two propositions.
\begin{proposition}\label{prop:oinfiqaou3we4hfrqi}
	Suppose $\sum_{j\in \{0, 1\}}\prob_F(Y_j = 1) > 0$. Then, $\Xi(F)$ is a one-sided interval of form $(\underline{\pi}, +\infty)$ or $[\underline{\pi}, +\infty)$.
\end{proposition}

Since $\Xi(F)$ is a subset of $\tilde{\Xi}(F)$, we know that $\tilde{\Xi}(F)$ also does not admit a finite upper bound.

\begin{proof}[Proof for Proposition \ref{prop:oinfiqaou3we4hfrqi}]
	Fix an arbitrary $F\in \mathcal{F}$ and fix an arbitrary parameter $\theta\in \Theta_I(F)$. 
	\begin{itemize}
	\item Let $\mathcal{H}_\theta$ be the set of all joint distributions $H$ for $(U, Z)$ with $Z\equiv (X, Y)$ that satisfies (\emph{i}) $\expt_{H}((U, Z)\in \Gamma(\theta)) = 1$; (\emph{ii}) $\expt_H [X (\indicator(U_j \le 0) - 0.5)] = 0$ for each $j\in \{0, 1\}$; (\emph{iii}) $H$'s marginal on $Z$ is equal to $F$. 
	\item Let $\mathcal{H}'_\theta$ be the set of all joint distributions $H'$ for $(U, X, Y, \tilde{Y})$ such that (\emph{i}) the marginal distribution of $H'$ for $(X, Y, U)\in \mathcal{H}_\theta$; (\emph{ii}) $(\tilde{Y}, U)\in \tilde{\Gamma}(Z, U;\theta)$ almost surely under $H'$ where $\tilde{\Gamma}(Z, U;\theta)$ is defined in \eqref{eq:entry_counterfactual_support}. 
	\item Let $\Xi_\theta$ be the set of all $\tilde{\theta}$ for which there exists a distribution $H'\in \mathcal{H}'_\theta$ such that $\tilde{\theta} = \expt_{H'} \sum_j \tilde{\pi}_j$ where $\tilde{\pi}_j = \tilde{Y}_j[\phi_j(X)^\top \alpha - \Delta_j \tilde{Y}_{1-j} + U_j]$.  
\end{itemize}

	Note that $\Xi(F) = \cup_{\theta\in \Theta_I(F)} \Xi_\theta$. Therefore, it suffices to prove that $\Xi_\theta$ is an one-sided interval with a finite bound on the lower end and an infinite bound on the upper end. Moreover, since $\mathcal{H}_\theta$ is a convex set, $\mathcal{H}'_\theta$ is also a convex set. Thus, $\Xi_\theta$ is a convex set, i.e., an interval in $\real$. Moreover, by the counterfactual support restriction in \eqref{eq:entry_counterfactual_support}, we know for any $\tilde{\theta}\in \Xi_\theta$, we must have $\tilde{\theta} \ge 0$. Thus, it suffices to show $\Xi_\theta$ is unbounded. Let's discuss two cases:
	\begin{itemize}
		\item Suppose there exists some $H'\in \mathcal{H}'_\theta$ such that for some $j^*\in \{0, 1\}$, $\rho \equiv \prob_{H'}(Y_{j^*} = 1, U_{j^*} > 0) > 0$. Let $\tilde{\theta}' \equiv \expt_{H'}\sum_j \tilde{\pi}_j$.  For any $M > 0$, we can construct the following $U^\dagger$ such that $U^{\dagger}_j = U_j$ if $j\neq j^*$ and, for $j = j^*$,
			\begin{equation*}
			U^{\dagger}_j = \begin{cases}
				U_j + M + \sum_{j'\in \{0, 1\}}|\Delta_{j'}| + |\phi(X)_j^\top \alpha - X_j^\top \alpha| & \text{if }Y_j = 1, U_j > 0\\
				U_j & \text{if otherwise}
			\end{cases}
			\end{equation*}
Let $H^{\dagger}$ denote the joint distribution of $(Y, X, U^{\dagger})$.  By construction, we know that (\emph{i}) $U^{\dagger}_{j^*} > 0$ if and only if $U_{j^*} > 0$, (\emph{ii}) $X_{j^*}^\top \alpha - \Delta_{j^*} Y_{1-j^*} + U^{\dagger}_{j^*} \ge 0$ if $Y_{j^*} = 1$.   Therefore, we know $H^\dagger$ belongs to $\mathcal{H}_\theta$. Since there always exist at least one pure-strategy Nash equilibrium in the counterfactual, we know there must exists some $H^*\in \mathcal{H}'_\theta$ such that $H^*$'s marginal on $(Y, X, U)$ is equal to $H^{\dagger}$. By the construction of $H^{\dagger}$, one can show that $\tilde{\theta}^* \equiv \expt_{H^*} \sum_j \tilde{\pi}_j \ge \tilde{\theta} + \rho M$. Since $\tilde{\theta}^* \in \Xi_\theta$ and $M$ is an arbitrary positive number, we know $\Xi_\theta$ does not have a finite upper bound.

\item Suppose, for the purpose of contradiction, there does not exist any $H'\in \mathcal{H}'_\theta$ such that for some $j\in \{0, 1\}$, $\prob_{H'}(Y_{j} = 1, U_{j} > 0) > 0$.  Since we know $\sum_{j \in \{0, 1\}} \prob_F(Y_j = 1) > 0$, there must exist some $j^*$ such that $\prob_F(Y_{j^*} = 1) > 0$. Pick an arbitrary $H \in \mathcal{H}_\theta$. Then, we know $\prob_H(Y_0 = 1, U_0 > 0) = 0$ and $\prob_H(Y_1 = 1, U_1 > 0) = 0$. Therefore, 
	\begin{equation}\label{eq:aoqifenf23}
\prob_H(Y_{j^*} = 1, U_{j^*} \le 0) = \prob_H(Y_{j^*} = 1)\quad \prob_H(Y_{j^*} = 0, U_{j^*} > 0) = \prob_H(U_{j^*} > 0)
\end{equation}
Define $p(x) \equiv \prob(U_{j^*} > 0 | X = x)$ and $q(x) \equiv \prob(Y_{j^*} = 1|X = x)$. Because of \eqref{eq:aoqifenf23}, we know $p(x) = \prob(Y_{j^*} = 0, U_{j^*} > 0| X = x)$ and $q(x) = \prob(Y_{j^*} = 1, U_{j^*} \le 0| X = x)$.

Because $\prob_H(Y_{j^*} = 1) > 0$ and $\prob_{H}[X (\indicator(U_{j^*} \le 0) - 0.5)]= 0 $, there must exist some set $A\subseteq \mathcal{X}$ with $\prob_H(X\in A) > 0$ such that $p(x) > 0$ and $q(x) > 0$ for almost every $x\in A$. 

Construct $\zeta$ as a random variable following uniform distribution in $[0, 1]$ and independent of $(Y, X, U)$. With a slight abuse of notation, let $H$ be the joint distribution of $(Y, X, U, \zeta)$. Then, construct another random variable $U^{\dagger}$ in the following way:

\begin{itemize}
\item for $X\notin A$, let $U^{\dagger} = U$.
\item for $X\in A$ and $p(X) \ge q(X)$, then set the value of $U^{\dagger}_j$ according to the following rule
	\begin{equation*}
	U^{\dagger}_j = \begin{cases}
		\max(42, -X_j^\top \alpha  + \sum_{j'\in \{0, 1\}}|\Delta_{j'}|)  &\text{if }Y_{j^*} = 1, U_{j^*} \le 0 \text{ and } j = j^*\\
		\min(-42, -X_j^\top \alpha - \sum_{j'\in \{0, 1\}}|\Delta_{j'}|)  & \text{if }Y_{j^*} = 0,\ U_{j^*} > 0, \zeta \le q(X) / p(X) \text{ and }j = j^*\\
		U_j & \text{if otherwise}
	\end{cases}
	\end{equation*}
\item for $X\in A$ and $p(X) < q(X)$, then set the value of $U^{\dagger}_j$ according to the following rule
	\begin{equation*}
	U^{\dagger}_j = \begin{cases}
		\max(42, -X_j^\top \alpha  + \sum_{j'\in \{0, 1\}}|\Delta_{j'}|)  &\text{if }Y_{j^*} = 1,U_{j^*} \le 0, \zeta \le p(X) / q(X) \text{ and } j = j^*\\
		\min(-42, -X_j^\top \alpha - \sum_{j'\in \{0, 1\}}|\Delta_{j'}|)  & \text{if }Y_{j^*} = 0,\ U_{j^*} > 0, \text{ and }j = j^*\\
		U_j & \text{if otherwise}
	\end{cases}
	\end{equation*}	
\end{itemize}
This constructed $U^{\dagger}$ satisfy the following conditions: (\emph{i}) $\prob(U_{j^*} \le 0|X) = \prob(U^\dagger_{j^*} \le 0|X)$ almost surely, and (\emph{ii}) $(Y, X, U^{\dagger})$ satisfies the support restriction in \eqref{eq:entry_game_support}. Thus, the distribution of $(Y, X, U^{\dagger})$, denoted as $H^{\dagger}$, must belongs to $\mathcal{H}_\theta$. By construction, $\prob_{H^\dagger}(Y_{j^*} = 1, U_{j^*} > 0) > 0$, which leads to a contradiction of the initial claim that there does not exist any $H'\in \mathcal{H}'_\theta$ such that for some $j\in \{0, 1\}$, $\prob_{H'}(Y_{j} = 1, U_{j} > 0) > 0$.
	\end{itemize}
	This completes the proof.
\end{proof}

\end{appendix}

\newpage

\printbibliography

\end{document}